\documentclass[12pt]{article}

\usepackage{graphicx}
\usepackage{amsmath}
\usepackage{amssymb}
\usepackage{amsthm}

\newcommand{\be}{\begin{equation}}
\newcommand{\ee}{\end{equation}}
\newcommand{\ba}{\begin{array}}
\newcommand{\ea}{\end{array}}
\newcommand{\bea}{\begin{eqnarray}}
\newcommand{\eea}{\end{eqnarray}}

\newcommand{\calH}{{\cal H }}

\newcommand{\calE}{{\cal E }}

\newcommand{\calZ}{{\cal Z }}

\newcommand{\calS}{{\cal S }}
\newcommand{\calG}{{\cal G }}

\newcommand{\calY}{{\cal Y }}
\newcommand{\calV}{{\cal V }}
\newcommand{\ZZ}{\mathbb{Z}}

\newcommand{\RR}{\mathbb{R}}

\newcommand{\la}{\langle}
\newcommand{\ra}{\rangle}

\newcommand{\nn}{\nonumber}

\newcommand{\trace}{\mathop{\mathrm{Tr}}\nolimits}

\newtheorem{dfn}{Definition}
\newtheorem{lemma}{Lemma}
\newtheorem{prop}{Proposition}
\newtheorem{theorem}{Theorem}

\newtheorem{corol}{Corollary}

\numberwithin{lemma}{section}
\numberwithin{corol}{section}
\numberwithin{prop}{section}
\numberwithin{dfn}{section}
\numberwithin{equation}{section}

\title{Topological quantum order: stability under local perturbations}

\author{Sergey Bravyi\footnote{IBM Watson Research Center, Yorktown Heights NY 10594 (USA); {\em sbravyi@us.ibm.com}},  \, \,
Matthew Hastings\footnote{Microsoft Research Station Q, CNSI Building, University of California, Santa Barbara, CA,
93106 (USA);  \mbox{\em \, \, \, \, mahastin@microsoft.com}}, \, \,
and  \, \,
Spyridon Michalakis\footnote{T-4 and CNLS, LANL - Los Alamos, NM, 87544 (USA); {\em spiros@lanl.gov}}
}

\begin{document}

\maketitle

\abstract{We study zero-temperature stability of topological phases of matter under weak time-independent perturbations.
Our results apply to quantum spin Hamiltonians that can be written  as a sum of geometrically local commuting projectors
on a $D$-dimensional lattice with certain topological order conditions.  Given such a Hamiltonian $H_0$ we prove  that  there exists a constant threshold $\epsilon>0$ such that for any
perturbation $V$ representable as a sum of short-range bounded-norm interactions the perturbed Hamiltonian
$H=H_0+\epsilon V$ has well-defined spectral bands originating from $O(1)$ smallest
eigenvalues of $H_0$. These bands are separated from the rest of the spectrum and from each other by a constant gap. The band originating from the smallest eigenvalue of $H_0$ has exponentially small width (as a function of the lattice size).

Our proof  exploits a discrete version of Hamiltonian flow equations,
the theory of relatively bounded operators, and the Lieb-Robinson bound.
 }

\newpage

\tableofcontents

\newpage

\section{Introduction}
The traditional classification of different phases of matter due to Landau rests on symmetry breaking.
Given a pair of gapped Hamiltonians $H_1,H_2$ with some symmetry group $G$,
the ground states  of $H_1$ and $H_2$ were considered to be in different phases if their symmetry breaking patterns are different. The discovery of topologically ordered phases, however, changes this paradigm.
Models such as Kitaev's toric code~\cite{tc} have ``topologically non-trivial" ground states
despite lacking any symmetry breaking. Such states cannot be changed into a ``topologically trivial" state such as
a product state  by any unitary locality-preserving operator~\cite{bhv}.

One possible approach to classifying topological phases is to call a pair of gapped Hamiltonians $H_1,H_2$
topologically equivalent iff  it is possible to connect $H_1$ and $H_2$ by a continuous path in the space of local gapped Hamiltonians. Using the idea of quasi-adiabatic continuation~\cite{hwen}, one can describe the evolution of the ground state
subspace along such a path by a unitary locality-preserving operator.
In particular ground state degeneracy and the geometry of ``logical operators" acting on the ground subspace is the same for
$H_1$ and $H_2$.

Most of the Hamiltonians describing TQO models such as
 Kitaev's quantum double model~\cite{tc} or Levin-Wen string-net model~\cite{LW} are not quite physical since they involve interactions affecting more than two spins at a time. One may hope however that such models emerge
 as low-energy effective Hamiltonians describing some simpler high-energy theories~\cite{Duan03,Koenig,BDLT08}.
 For example, the toric code model
 with four-spin interactions can be ``implemented" as the fourth-order effective Hamiltonian describing low-energy limit
 of the honeycomb model~\cite{Kitaev05} which involves only two-spin interactions. The higher-order corrections to the
 effective Hamiltonian must be regarded as a perturbation. Thus in order to show that the honeycomb model is topologically
equivalent to the toric code (in the $J_z\gg J_x,J_y$ phase) one has to prove that the spectral gap in the toric code Hamiltonian
does not close in a presence of weak perturbations $V$ that can be represented as a sum of bounded-norm
short-range (exponentially decaying) interactions.

Even if one leaves aside the question of how multi-spin interactions can be implemented in a lab, one has to worry
about precision up to which an ideal model Hamiltonian can be approximated in a real life.
If the presence or absence of the gap depends on tiny variations of the Hamitonian parameters that are beyond
experimentalist's control, the distinction between gapped and gapless Hamiltonians is meaningless.
The best we can hope for is to approximate individual interactions of the ideal model with some constant
precision $\epsilon$  independent of the system size $N$. Accordingly, the ideal Hamiltonian can be approximated only up to an extensive error $O(\epsilon N)$. Proving stability of topological phases
thus reduces to proving that the spectral gap of the ideal TQO models does not close in the presence
of such extensive perturbations.

Currently, the tools for proving lower bounds on the spectral gap are fairly limited.  For example, one of the outstanding problems
in mathematical physics is to prove the existence of a spectral gap for the spin-$1$ Heisenberg chain, making rigorous the arguments of Haldane~\cite{haldane}.  Some progress toward this was obtained by Yarotsky~\cite{yarot}, who showed the stability of the gap near
the AKLT point~\cite{aklt}.
Yarotsky's tools however are limited to  perturbations of Hamiltonians which
are topologically trivial.  Thus, new methods are needed to analyze topologically ordered phases.  Some partial
results were recently obtained by Trebst et al~\cite{trebst} and Klich~\cite{klich} who proved gap stability for the toric code under a special type of perturbations diagonal in the $z$-basis as well as for anyon lattices on a sphere.

In the present paper we succeed in proving gap stability under  generic local perturbations.  Our results are valid not just for the toric code, but more generally for any Hamiltonian which can be written as a sum of geometrically
local commuting projectors on a $D$-dimensional lattice with certain topological order conditions that we define
later.  This includes models such as Kitaev's quantum double model~\cite{tc} and the Levin-Wen string-net model~\cite{LW}.
Furthermore, we prove stability of the spectral gaps separating sufficiently low-lying eigenvalues
of the unperturbed Hamiltonian. In the case of 2D models with anyonic excitations it allows us
to define string-like operators that create particle excitations for the perturbed Hamiltonian
and prove stability of invariants describing the braiding statistics of excitations.
We explain how this may be used to adiabatically control a perturbed topological model to perform braiding
operations to manipulate topologically protected quantum information.

\subsection{Summary of results}
\label{subs:summary}
Consider a system composed of  finite-dimensional quantum particles (qudits)
occupying sites of a $D$-dimensional lattice $\Lambda$ of linear size $L$.
The corresponding Hilbert space is a tensor product of the local Hilbert spaces,
$\calH=\bigotimes_{u\in \Lambda} \calH_u$,
$\dim{(\calH_u)}=O(1)$.
Suppose the unperturbed Hamiltonian $H_0$  can be written as a sum of
geometrically local pairwise commuting projectors,
\[
H_0=\sum_{A \subseteq \Lambda} Q_A,
\]
where the sum runs over all subsets of the lattice of diameter $O(1)$
and $Q_A$ is a projector acting non-trivially only on sites of $A$
(one may have $Q_A=0$ for some subsets $A$). The commutativity assumption
implies that all projectors $Q_A$ can be diagonalized in the same basis.
Accordingly, all eigenvalues of $H_0$ are non-negative integers.
We assume that the smallest eigenvalue of $H_0$ is zero, that is,
ground states of $H_0$ are annihilated by every projector $Q_A$.
Such states span the ground subspace $P$,
\[
P=\{ |\psi\ra\in \calH  \, : \, Q_A \, |\psi\ra =0 \quad \mbox{for all $A\subseteq \Lambda$}\}.
\]
For any subset $B\subseteq \Lambda$ we shall also define a local ground subspace as
\[
P_B = \{ |\psi\ra\in \calH \, : \, Q_A \, |\psi\ra =0 \quad \mbox{for all $A\subseteq B$}\}.
\]
We shall use the notations $P$ and $P_B$ both for linear subspaces and for the corresponding projectors.
Note that the projector $P_B$ acts non-trivially only on the subset $B$.

We shall impose two extra conditions on $H_0$ and the ground subspace $P$ that guarantee the gap stability.
Let us first state these conditions informally (see Section~\ref{subs:tqo} for formal definitions):

\begin{enumerate}
\item[] {\bf TQO-1:} The ground subspace $P$ is a quantum  code with a macroscopic distance\footnote{For our purposes
it suffices that the distance grows as a positive power of the lattice size $L$},
\item[] {\bf TQO-2:} Local ground subspaces are consistent with  the global one
\end{enumerate}
Condition TQO-1 is the traditional definition of TQO. It guarantees that a local operator cannot induce transitions
between orthogonal ground states or distinguish a pair of orthogonal ground states from each other.
Thus a local perturbation can lift the ground state degeneracy only in
the $n$-th order of perturbation theory, where $n$ can be made arbitrarily large by increasing the lattice size, see~\cite{tc}.
Surprising, condition TQO-1 by itself is not sufficient for stability, see a simple counter-example in Section~\ref{subs:unstable}.

Condition TQO-2  demands  that a local ground subspace
$P_B$ and the global ground subspace $P$ must be
consistent, namely, the projectors $P_B$ and $P$ must have the same
reductions on any subset $A\subset B$ which is "sufficiently far" from the boundary of $B$.
We need to impose TQO-2 only for regions with trivial topology such a cube or a ball.
The consistency between the global and the local ground subspaces
may be violated for regions with non-trivial topology. For example, if $B$ has a hole,
the local ground subspace $P_B$ may include sectors with a non-trivial topological charge
inside the hole as opposed to the global ground subspace.
Condition TQO-2   by itself is also not sufficient for stability, see a counter-example in Section~\ref{subs:unstable}.

Let us emphasize  that all our results apply also  to the special case when $H_0$ has non-degenerate ground
state. In this case TQO-1 is automatically satisfied since $P$ is a one-dimensional subspace and thus condition TQO-2 alone guarantees the gap stability.

We consider a perturbation $V$ that can be written as a sum of bounded-norm interactions
\[
V=\sum_{r\ge 1}\; \;  \sum_{A\in \calS(r)}  \; \; V_{r,A},
\]
where $\calS(r)$ is a set of cubes of linear size $r$ and $V_{r,A}$ is an operator acting on sites of $A$.
We assume that the magnitude of interactions decays exponentially for large $r$,
\[
\max_{A\in \calS(r)} \|V_{r,A}\|\le Je^{-\mu r},
\]
where $J,\mu>0$ are some constants independent of $L$.
Our main result is the following.
\begin{theorem}
Suppose $H_0$ obeys TQO-1,2. Then there exist constants $J_0,c_1,c_2>0$ depending only on $\mu$
and the spatial dimension $D$ such that for all $J\le J_0$ the spectrum of $H_0+V$ is
contained (up to an overall energy shift) in the union of intervals $\bigcup_{k\ge 0} I_k$, where
$k$ runs over the spectrum of $H_0$ and
\[
I_k=\{\lambda\in \RR\, : \, k(1-c_1J)-\delta\le \lambda\le k(1+c_1J)+\delta \},
\]
and
\[
\delta =poly(L) \exp{(-c_2 L^{3/8})}.
\]
\label{thm:main}
\end{theorem}
In other words, the perturbation $V$ changes positive eigenvalues of $H_0$ at most by a constant
factor $1\pm c_1J$ (neglecting the exponentially small correction $\delta$) while
the smallest eigenvalue $k=0$ is transformed into a band $I_0$ of exponentially small width $2\delta$,
see Fig.~\ref{fig:energy}.
%Note that the bands $I_k$ and $I_m$ with $k\ne m$ are separated by a gap at least $1/2$
%provided that
%\[
%J<\frac{|m-k|-1/2}{c_1(m+k)}.
%\]
One can easily check that
for any fixed $k$ the band $I_k$ is separated from all other bands $I_m$, $m\ne k$, by a gap at least
$1/2$ provided that $J<J_k$, where
\[
J_k=\frac1{c_1(4k+2)}.
\]
Thus the bands originating from eigenvalues $0,1,\ldots,k$ of $H_0$ are separated from each other
and from the rest of the spectrum by a gap at least $1/2$ provided that $J<J_k$.

In the case when excitations of $H_0$ are anyons, one can infer all topological invariants
such as $S$, $R$, and $F$-matrices by evaluating fusion and braiding diagrams with only a few
particles (for example $4$ particles suffice to compute all $F$ matrices).
Accordingly, any matrix element of, say, $F$-matrix, can be represented as an expectation value
$\la \psi_0|O_m\ldots O_2 O_1|\psi_0\ra$ where $|\psi_0\ra$ is the ground state and $O_i$
are operators creating pairs of excitations from the ground state, moving, fusing, and annihilating them.
Our stability result for excited states with $O(1)$ excitations allows us to
construct quasi-adiabatic continuation of operators $O_i$ thus explicitly demonstrating that the perturbed Hamiltonian
has the same $S,R$, and $F$ matrices as the ideal one, see Section~\ref{sec:continue}.

\begin{center}
\begin{figure}[htb]
\centerline{
\mbox{
 \includegraphics[height=4cm]{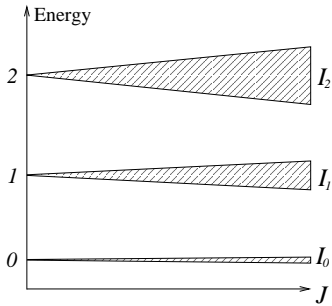}
 }}
 \caption{Energy bands $I_k$ describing the spectrum  of a perturbed Hamiltonian $H_0+V$.}
 \label{fig:energy}
\end{figure}
\end{center}

\subsection{Sketch of the stability proof}
Let us sketch the main steps of the proof of Theorem~\ref{thm:main}.
We start from proving the theorem for a special class of perturbation $V$
such that all individual  interactions $V_{r,A}$ preserve the ground subspace $P$,
that is, $[P,V_{r,A}]=0$.
We call such perturbations {\em block-diagonal}. In Section~\ref{sec:rel} we prove that
block-diagonal perturbations are {\em relatively bounded} by $H_0$, that is,
$\| V \psi\|\le b \|H_0\psi\|$ for any state $|\psi\ra\in \calH$ and for some coefficient
$b=O(J)$. Here for simplicity we ignore some exponentially small corrections.
A nice feature of relatively bounded perturbations is that
the spectrum of a perturbed Hamiltonian $H_0+V$ is contained in the union
of intervals $I_k$ where $k$ runs over the spectrum of $H_0$ and $I_k=(k(1-b),k(1+b))$,
see Section~\ref{subs:rel}.
The proof of the relative boundness is rather elementary and uses certain decomposition of the Hilbert space
in terms of {\em syndrome subspaces} which is a standard tool in the theory of quantum error correcting codes.
In order to get a strong enough bound on the coefficient $b$ we use a novel technique of ``coarse-graining"
the syndrome subspaces, see Section~\ref{subs:bd_stability} for details.

In the second part of the proof we reduce generic perturbations $V$ to block-diagonal perturbations.
Specifically, we construct a unitary operator $U$ such that $U(H_0+V)U^\dag\approx H_0+W$,
where $W$ is a block-diagonal perturbation. Since $U$ does not change eigenvalues,
we can use the techniques described above to analyze the spectrum of $H_0+V$.
The operator $U$ is constructed using a discrete version of {\em Hamiltonian flow equations}
developed by Glazek, Wilson, and Wegner~\cite{flow}.
Specifically, we define a hierarchy of  Hamiltonians $H(n)=H_0+V(n) +W(n)$
labeled by an integer level $n\ge 0$, such that $W(n)$ is a block-diagonal perturbation
while $V(n)$ is a generic perturbation. We start at the level $n=0$ with the perturbed  $H_0+V$, that is,
$V(0)=V$ and $W(0)=0$. As we go to higher levels, the Hamiltonian $H(n)$ becomes more close
to a block-diagonal form.
The transformation from $H(n)$ to $H(n+1)$ is described
by a unitary operator $U(n)$ that block-diagonalizes $H(n)$ up to errors of order $V(n)^2$.
These errors are dealt with at the next level of the hierarchy, see Section~\ref{sec:flow} for details.
We construct $U(n)$ by solving a linearized block-diagonalization problem, see
Section~\ref{sec:lbd}. The solution can be easily constructed in terms of the series while convergence of the
series follows from the fact that $W(n)$ is relatively bounded by $H_0$.

We prove that the strength of $V(n)$ decays doubly-exponentially as a function of $n$, while
$W(n)$ does not change essentially after the first few levels.
We then choose the desired unitary operator $U$ as $U=U(n_{f})\cdots U(1) U(0)$
where the highest level $n_{f}\sim \log{(L)}$ is chosen to make the norm of $V(n_f)$ exponentially small (as a function of $L$).
The most technical  part of the proof is to show that the unitary operators $U(n)$ are locality preserving
such that all Hamiltonians $V(n)$ and $W(n)$ remain sufficiently local.
To this end we first prove that $U(n)$ can be generated by a quasi-local Hamiltonian, see Section~\ref{sec:lbd},
and then employ the Lieb-Robinson bound, see Section~\ref{sec:LR}.

\section{Hamiltonians describing TQO}
\label{sec:tqo}
\subsection{Frustration-free commuting Hamiltonians}
\label{subs:ffc}
To simplify notations we shall restrict ourselves to the spatial dimension $D=2$.
A generalization to an arbitrary $D$ is straightforward.
Let $\Lambda=\ZZ_L\times \ZZ_L$ be a two-dimensional square lattice of linear size $L$ with periodic boundary conditions.
We assume that every site $u\in \Lambda$ is occupied by a finite-dimensional quantum particle (qudit)
such that the Hilbert space describing $\Lambda$ is a tensor product
\[
\calH=\bigotimes_{u\in \Lambda} \calH_u, \quad \dim{\calH_u}=O(1).
\]
Let $\calS(r)$ be a set of all square blocks $A\subseteq \Lambda$ of size $r\times r$,
where $r$ is a positive integer.
Note that $\calS(r)$ contains $L^2$  translations of some elementary square
of size $r\times r$ for all $r<L$, $\calS(L)=\Lambda$, and $\calS(r)=\emptyset$ for $r>L$.
We can always assume that the unperturbed Hamiltonian $H_0$ involves only $2\times 2$ interactions
(otherwise consider a coarse-grained lattice):
\be
\label{H_0}
H_0=\sum_{A\in \calS(2)} G_A.
\ee
There will be three essential restrictions on the form of interactions $G_A$.
Firstly, we require that $G_A$ are {\em pairwise commuting} operators, that is,
\[
G_A \, G_B = G_B \, G_A \quad \mbox{for all $A,B\in \calS(2)$}.
\]
Thus all interactions $G_A$ can be diagonalized in the same basis.
Secondly, we require that $H_0$ is a {\em frustration free} Hamiltonian, that is,
the ground state of $H_0$ minimizes energy of every individual term $G_A$.
Performing an overall energy shift we can always assume that all $G_A$ are positive-semidefinite
operators,
\[
G_A \ge 0.
\]
Then the condition of being frustration-free demands that ground states of $H_0$ are
common zero eigenvectors of every term $G_A$. Thus the ground subspace of $H_0$ is
\be
P=\{ |\psi\ra \in \calH \, : \, G_A\, |\psi\ra =0 \quad \mbox{for all $A\in \calS(2)$}\}.
\ee
Thirdly, we shall assume that every operator $G_A$ has a {\em constant spectral gap}, that is,
the smallest positive eigenvalue of $G_A$ is bounded from below by a constant independent of the lattice size $L$.
We can always normalize the Hamiltonian $H_0$ such that the spectral gap of any $G_A$ is at least $1$.
This is equivalent to a condition
\[
G_A^2\ge G_A.
\]
Let $P_A$ be the projector onto the zero subspace of $G_A$  and $Q_A=I-P_A$.
Note that all the projectors $P_A,Q_A$ are pairwise commuting.
For any square $B\in \calS(r)$, $r\ge 2$ define a projector
onto the local ground subspace
\be
\label{local_projectors}
P_B =\prod_{\substack{A\in \calS(2) \\ A\subseteq B}} P_A
\ee
and $Q_B=I-P_B$.
Note that $P_B$ and $Q_B$ have support on  $B$.
We shall often use the same notation for a subspace and for the corresponding projector.

\subsection{Formal definition of TQO}
\label{subs:tqo}
We shall need two extra property of $H_0$ and the ground subspace $P$ that guarantee the gap stability
and robustness of the ground state degeneracy.
 We shall assume that there exists a constant $c>0$
such that the following conditions hold for some integer $L^*\ge cL$
for all sufficiently large $L$:

\begin{center}
\begin{tabular}{|c|}
\hline \\
{\bf TQO-1:} \parbox[t]{13.5cm}{Let $A\in \calS(r)$ be any square of size $r\le L^*$. Let $O_A$
be any operator acting on $A$. Then
\[
PO_AP =c P
\]
for some complex number $c$.} \\  \\
{\bf TQO-2:} \parbox[t]{13.5cm}{Let $A\in \calS(r)$ be any square of size $r\le L^*$
and let $B\in \calS(r+2)$ be the square that contains $A$ and all nearest neighbors of $A$.
Define reduced density matrices $\rho_A = \trace_{A^c} (P)$ and
$\rho_A^{(B)}=\trace_{A^c} (P_B)$.
Then the kernel of $\rho_A$ coincides with the kernel of $\rho_A^{(B)}$.
}\\ \\
\hline
\end{tabular}
\end{center}

{\em Remark 1.} Using the language of quantum error correcting codes
one can define the {\em minimum distance} of $P$ as
the smallest integer $d$ such that
erasure of any subset of $d$ particles can be corrected
for any encoded state $|\psi\ra \in P$,  see~\cite{BPT09} for details.
Note that TQO-1 holds for $L^*=\lfloor \sqrt{d} \rfloor$ since
the reduced state of any square $A\in \calS(L^*)$ does not depend on the encoded
state. What is less trivial, TQO-1 holds also for $L^*=\Omega(d)$, see~\cite{BPT09}.
Thus $L^*$ coincides with the distance of the code $P$ up to a constant coefficient
(as far as condition TQO-1 is concerned).

{\em Remark 2.} Condition TQO-2 can be easily `proved' if the  excitations of $H_0$
are anyons (since the latter assumption lacks a rigorous formulation, the argument given below is not completely
rigorous either).
 Indeed, in this case we can choose a complete basis of the excited subspace $Q$ such that the basis vectors
correspond to various configurations of anyons.  For non-abelian theories one may have several basis vectors
for a fixed configuration of anyons that describe different fusion channels, see~\cite{Kitaev05}.
Note that any state $|\psi\ra\in P_B$ is a superposition of configurations with no anyons inside $B$.
Since $A$ is a topological trivial region, any such configuration can be prepared from the vacuum
$P$ by some unitary operator $U_{A^c}$ acting on complementary region  $A^c=\Lambda\backslash A$.
Thus $|\psi\ra= U_{A^c}|\psi_0\ra$ for some ground state $|\psi_0\ra\in P$.
Since all ground states $|\psi_0\ra$ have the same reduced matrix on $A$, it means that $|\psi\ra$
and $P$ have the same reduced matrix on $A$. This implies TQO-2.
The above arguments suggest that  TQO-2 holds for all 2D models of  TQO that can be
described by commuting frustration-free
such as quantum double models~\cite{tc} and Levin-Wen string-net models~\cite{LW}.

{\em Remark~3.} As was already mentioned, the consistency between the global and the local ground subspaces
may be violated for regions with non-trivial topology. For example, if $A$ has a hole,
the local ground subspace $P_A$ may include sectors with a non-trivial topological charge
inside the hole as opposed to the global ground subspace.

We shall need the following corollary of TQO-2.
\begin{corol}
\label{corol:tqo3}
Let $A\in \calS(r)$ be any square of size $r\le L^*$ and $O_A$ be any operator acting on $A$
such that $O_A P=0$. Let $B\in \calS(r+2)$ be the square that contains $A$ and all nearest neighbors of $A$.
Then $O_A P_B=0$.
\end{corol}
\begin{proof}
Let $\rho_A = \trace_{A^c} (P)$. The assumption $O_A P O_A^\dag=0$
implies that $O_A \rho_A O_A^\dag=0$, that is, $O_A$ annihilates any state in the range of $\rho_A$.
From TQO-2, the  range of $\trace_{A^c}(P_B)$ coincides with the range of $\rho_A$, and thus $\trace(O_A P_B O_A^\dag)=0$.
It implies $O_A P_B=0$.
\end{proof}

\subsection{Verification of TQO conditions  for stabilizer Hamiltonians}
\label{subs:stabilizer}
Conditions TQO-1,2 can be easily checked for
those models of TQO that can be described using the stabilizer formalism
such as the toric code model~\cite{tc} or topological color codes~\cite{Bombin06}.
For such models each site of the lattice $\Lambda$ represents one or several qubits, while
 the ground state subspace $P$
is a stabilizer code, i.e., the invariant subspace of some abelian stabilizer group $\calG\subseteq \mathrm{Pauli}(\Lambda)$.
Here $\mathrm{Pauli}(\Lambda)$ is a group generated by single-qubit Pauli operators $\sigma^x_i,\sigma^y_i,\sigma^z_i$.
The stabilizer group must have a set of geometrically local generators, that is, $\calG=\la S_1,\ldots,S_M\ra$
where any generator $S_a \in \mathrm{Pauli}(\Lambda)$ acts non-trivially only on $O(1)$ qubits located within distance $O(1)$
from each other. Note that the generators need not to be independent.
We choose the  corresponding stabilizer Hamiltonian $H_0$ as
\[
H_0=\sum_a (I-S_a)/2
\]
such that states invariant under action of stabilizers have zero energy.
The {\em minimal distance} of the code is the smallest integer $d$ such that there exists
a Pauli operator $O$ that commutes with all elements of $\calG$ but does not belong to $\calG$.
Such an operator $O$ can be regarded as a logical Pauli operator acting on encoded states.
It follows from results of~\cite{BPT09} that condition TQO-1 holds if we choose $L^*=\Omega(d)$.

Assume that the set of qubits is coarse-grained  into sites of the lattice
 $\Lambda$ such that the support of any generator $S_a$
is contained in at least one $2\times 2$ square.
One can bring this Hamiltonian into the form Eq.~(\ref{H_0}) by distributing the generators over $2\times 2$
squares in an arbitrary way.  For any  square $B\in \calS(r)$ one can define two subgroups of $\calG$:
(i) a subgroup $\calG_B$
generated by generators $S_a$ whose support is contained in $B$, and
(ii) a subgroup $\calG(B)$ that
 includes all stabilizers
$S\in \calG$ whose support is contained in $B$.
By definition, $\calG_B \subseteq \calG(B)$, but in general $\calG_B \ne \calG(B)$.
\begin{lemma}
The stabilizer Hamiltonian $H_0$ satisfies condition TQO-2 iff
for any square $A\in \calS(r)$, $r\le L^*$, one has
$\calG(A) \subseteq \calG_B$, where $B=b_1(A)$.
\end{lemma}
Thus  TQO-2 demands that  any
element  of the stabilizer group whose support is contained in a square  $A$
 can be written as a product of generators whose support is contained in $A$ and a small
 neighborhood of $A$.
 We leave verification of this condition for the toric code model as an exercise for the reader.

\begin{proof}
Indeed, the reduced density matrix $\rho_A$ computed using the global
ground subspace $P$ is proportional to the projector
onto the codespace of the stabilizer code $\calG(A)$.
The reduced density matrix $\rho_A$ computed using
the local ground subspace $P_B$ is proportional
onto the codespace of the stabilizer code $\calG_B(A)$,
where $\calG_B(A)$ includes all elements of $\calG_B$
whose support is contained in $A$.
Thus TQO-2 holds iff $\calG(A) = \calG_B(A)$.
This is equivalent to the condition of the lemma.
\end{proof}

\subsection{Unstable version of the toric code model}
\label{subs:unstable}
In this section we demonstrate that condition TQO-1 alone is not sufficient
for stability. Let us start from the standard toric code model~\cite{tc},
\[
H_{tc} = -\sum_p B_p - \sum_s A_s,
\]
where qubits live on edges of a square 2D  lattice, $p$ and $s$ labels plaquettes and
sites of the lattice, $B_p$ is a product of $\sigma^z$ over the four boundary edges of $p$,
and $A_s$ is a product of four $\sigma^x$ over the four edges incident to $s$.
We shall refer to $B_p$ and $A_s$ as plaquette and star operators.
The ground subspace $P$ is defined by eigenvalue equations $B_p=1$
for all $p$ and $A_s=1$ for all $s$. It is well known that
$P$ is a quantum code with the minimal distance $d=L-1$. Hence $P$ obeys TQO-1,2 with
$L^*=L-1$.

Consider now  a modified toric code model
\[
H_{tc}'=-\sum_{(p,q)} B_p B_q - \sum_{s} A_s  - B_{p^*}
\]
where $(p,q)$ labels pairs of adjacent plaquettes and
$p^*$ is some selected plaquette. We assume that the total
number of plaquettes $N_p$ is even.
Note that $H_{tc}'$ is a frustration-free commuting Hamiltonian.
In addition, $H_{tc}'$ and $H_{tc}$ have the same ground subspace $P$
corresponding to $B_p=1$, $A_s=1$ for all $s$ and $p$.
Hence $H_{tc}'$ obeys TQO-1. We claim that $H_{tc}'$ violates TQO-2.
Indeed, choose any square $A$ located sufficiently far from the selected plaquette $p^*$.
Then the local ground subspace $P_A$ has equal contributions from sectors
$B_p=1$, $A_s=1$ and $B_p=-1$, $A_s=1$ thus being inconsistent with
the global ground state.

Let us now argue that the spectral gap of $H_{tc}'$ closes in a presence of local perturbations
with a strength of order $1/N_p$. This instability has the same origin as the instability of the classical
2D Ising model under external magnetic field.
Note that $H_{tc}'$ has spectral gap $\Delta=2$ and the second smallest eigenvalue belongs to the sector
$A_s=1$, $B_p=-1$ for all $s$ and $p$.
Consider a perturbation describing an ``external magnetic field",
\[
V=h\sum_p B_p, \quad h>0.
\]
For sufficiently large $h$, say, $h=4/N_p$, the ground state of $H_{tc}'+V$  moves from the sector $A_s=1,B_p=1$ to the sector
$A_s=1$, $B_p=-1$. Hence the gap above the ground state closes for some intermediate value of $h$.

Needless to say, condition TQO-2 alone is also not sufficient for stability.
The simplest counter-example is the 2D classical Ising model in which the gap is unstable under external magnetic field.

\section{Relatively bounded perturbations}
\label{sec:rel}

\subsection{Definition and basic properties}
\label{subs:rel}
In this section we introduce necessary facts from the theory of relatively bounded perturbations.
It mostly follows Chapter~IV of~\cite{Kato} although our definitions and proofs are much simpler since
we are interested  only in finite-dimensional Hilbert spaces.

Let $H_0$ and $W$ be any Hamiltonians acting on some Hilbert space $\calH$. We shall say that $W$ is {\em relatively bounded} by $H_0$
iff there exist $0\le b<1$  such that
\be
\label{rel_bound}
\| W \psi \| \le b\,  \| H_0 \psi\| \quad \mbox{for all $|\psi\ra \in \calH$}.
\ee
The notion  of a relatively bounded perturbation allows one to define a ``weak perturbation" and
rigorously justify application of perturbative expansions
even when the norm of $W$ is much larger than the spectral gap of $H_0$.
We shall be mostly interested in the case when
$b$ is a constant independent of the lattice size $L$.
Note that the condition Eq.~(\ref{rel_bound}) is
equivalent to an operator inequality $W^2 \le b^2 H_0^2$.

The following lemma asserts that a relatively bounded perturbation can change
eigenvalues of $H_0$ at most by a factor $1\pm b$.
\begin{lemma}
\label{lemma:relb}
Suppose $W$ is relatively bounded by $H_0$.
Then the spectrum of $H_0+W$ is contained in the union
of intervals $[\lambda_0(1-b),\lambda_0(1+b)]$ where $\lambda_0$
runs over the spectrum of $H_0$.
\end{lemma}
\begin{proof}
Indeed, suppose $(H_0+W)\, |\psi\ra = \lambda\, |\psi\ra$, that is,
\be
(H_0-\lambda \, I) \, |\psi\ra = -W\, |\psi\ra.
\ee
The relative boundness then implies $\| (H_0-\lambda \, I) \psi\| \le b \| H_0 \psi\|$, that is,
\be
\label{aux00}
\la \psi| (H_0-\lambda \, I)^2 |\psi\ra  \le b^2 \la \psi| H_0^2  |\psi\ra.
\ee
Let $H_0=\sum_{\lambda_0} \lambda_0 P_{\lambda_0}$ be the spectral decomposition of $H_0$.
Here the sum runs over the spectrum  of $H_0$ and $P_{\lambda_0}$ is a projector
onto the eigenspace with an  eigenvalue $\lambda_0$.
Define a probability distribution $p(\lambda_0)=\la \psi|P_{\lambda_0} |\psi\ra$.
Substituting it into Eq.~(\ref{aux00}) one gets
\be
\sum_{\lambda_0} (\lambda_0-\lambda)^2  \, p(\lambda_0) \le \sum_{\lambda_0} b^2 \lambda_0^2 \, p(\lambda_0).
\ee
Therefore there exists at least one eigenvalue $\lambda_0$ such that
\be
(\lambda_0-\lambda)^2 \le b^2 \lambda_0^2.
\ee
This is equivalent to $\lambda_0(1-b)\le \lambda\le \lambda_0(1+b)$.
\end{proof}

\subsection{Stability of TQO under block-diagonal perturbations}
\label{subs:bd_stability}

In this section we shall consider perturbations
\[
W=\sum_{A\in \calS(q)} W_A
\]
such that all local terms $W_A$ are block-diagonal,
\[
[W_A,P]=0 \quad \mbox{for all $A\in \calS(q)$}.
\]
We shall assume that $q\le L^*$, so
condition TQO-1 implies that  the restriction
of $W_A$ onto the $P$-subspace is a multiple of the identity.
Since we are not interested in the overall shift in energy, we can assume that
\be
\label{bd2}
W_A\, P =0 \quad \mbox{for all $A\in \calS(q)$}.
\ee
The interaction strength of $W$ will be measured by a parameter
\be
\label{w}
w=\max_{A\in \calS(q)} \|W_A\|.
\ee

\begin{lemma}
\label{lemma:bd_stability}
Let $W$ be a perturbation satisfying Eqs.~(\ref{bd2},\ref{w}).
Then $W$ is relatively bounded by  $H_0$ with a constant
\[
b=O(w q^2).
\]
\end{lemma}
\begin{proof}
Let us start from introducing some notations.
A {\em syndrome} $s\, : \, \calS(2) \to \{0,1\}$  is a function that assigns an eigenvalue $s_A\in \{0,1\}$ to every projector $Q_A$, $A\in \calS(2)$.
Given a syndrome $s$ and a square $A\in \calS(2)$ we shall say that $A$ is a {\em defect}
iff $s_A=1$. Thus one can consider $s$ as a configuration of defects.
For any syndrome $s$  define a projector
\[
R_s = \prod_{A\in \calS(2)} \left[ s_A Q_A + (1-s_A) (I- Q_A)\right]
\]
projecting onto a subspace spanned by states with a syndrome $s$.
Clearly the family of projectors $R_s$ defines an orthogonal decomposition
of the Hilbert space, that is, $\sum_s R_s=I$.

%SBB5: simplified proof + more explanations
Let us fix some
partition of the lattice into contiguous $q\times q$ squares $B_1,\ldots,B_M \in \calS(q)$
(if $L$ is not a multiple of $q$, the squares $B_i$ may have size  $q\pm O(1)$).
We shall refer to a set of $2\times 2$ squares
contained in a particular square $B_j$ as a {\em box}. We shall need the following properties:
\begin{enumerate}
\item  Every square $A\in \calS(2)$ is contained in  exactly one box $B_i$,
\item Each box $B_i$ overlaps with $O(q^2)$ squares $C\in \calS(q)$,
\item Each square $C\in \calS(q)$ overlaps with $O(1)$ boxes $B_i$.
\end{enumerate}
Given a syndrome $s$ and a box $B_i$ we shall say that $B_i$ is {\em occupied}
if $B_i$ contains at least one defect, that is, there is a $2\times 2$ square $A\subset B$ such that $s_A=1$.
Otherwise we shall say that the box $B_i$ is {\em empty}.

Given a syndrome $s$ let $b(s)\subseteq [M]$ be the subset of occupied boxes.
For any subset of boxes $\calY \subseteq [M]$ define a projector
\[
R_\calY=\sum_{s \, : \, b(s)=\calY} R_s.
\]
It projects onto the subspace in which all boxes in $\calY$ are occupied and the remaining
boxes are empty. Clearly, the family of projectors $R_\calY$ defines
an orthogonal decomposition, that is, $\sum_{\calY \subseteq [M]} R_\calY=I$.
We claim that
any operator $W_A$ acting on a square $A\in \calS(q)$ and  satisfying Eq.~(\ref{bd2}) has only a few
off-diagonal blocks with respect to this decomposition. Specifically, Corollary~\ref{corol:tqo3}
 implies that
\be
\label{trans}
R_\calY W_A R_\calZ\ne 0
\ee
only if $A$ has distance $O(1)$ from some occupied box in $\calY$ {\em and}
$A$ has distance $O(1)$ from some occupied box in $\calZ$, {\em and}
the configurations $\calY,\calZ$ differ only at those boxes that overlap with $A$.
Clearly, for any fixed $\calY\subseteq [M]$ such that $\calY$ has $k$ occupied boxes
the number of pairs $(A\in \calS(q), \calZ\subseteq [M])$ that could satisfy Eq.~(\ref{trans})
is at most $O(kq^2)$.
Thus for any state $|\psi\ra$ we get
\bea
\la \psi|W^2|\psi\ra &= &  \sum_{\calY,\calZ,\calV\subseteq [M]} \la \psi|R_\calY W R_\calZ W R_\calV |\psi\ra   \nn \\
&\le &
\sum_{\calY,\calZ,\calV\subseteq [M]} \| R_\calY W R_\calZ \| \cdot \| R_\calZ W R_\calV \| \cdot
\| R_\calY \psi \| \cdot \| R_\calV \psi\| \nn \\
&\le &
\sum_{\calY,\calZ,\calV\subseteq [M]} \| R_\calY W R_\calZ \| \cdot \| R_\calZ W R_\calV \| \cdot \frac12 (\la \psi|R_\calY |\psi\ra+
\la \psi|R_\calV|\psi\ra )\nn \\
&= & \sum_{\calY,\calZ,\calV\subseteq [M]}  \| R_\calY W R_\calZ \| \cdot \| R_\calZ W R_\calV \| \cdot \la \psi|R_\calY |\psi\ra \nn \\
& \le & \sum_{k\ge 0}\; \;  \sum_{\calY\, : \, |\calY|=k} \; \; O(k^2 q^4 w^2) \la \psi|R_\calY |\psi\ra = O(w^2 q^4 ) \la \psi|G|\psi\ra,
\label{W^2}
\eea
where
\be
G=\sum_{k\ge 0} \; \; \sum_{\calY\, : \, |\calY|=k} \; \; k^2 R_\calY.
\ee
The  inequality Eq.~(\ref{W^2}) follows from the fact that $\calY$ and $\calZ$ differ at at most $O(1)$ boxes
and an obvious bound $k(k+O(1))=O(k^2)$. Finally, note that $G\le H_0^2$ since
any configuration of defects with $k$ occupied boxes must have at least $k$ defects
and since creating a defect costs at least a unit of energy.
We arrive at
\be
\la \psi|W^2|\psi\ra \le b^2 \la \psi|H_0^2 |\psi\ra, \quad b=O(wq^2).
\ee
It completes the proof.
\end{proof}

%SBB1: new stuff
We shall also need a local version of Lemma~\ref{lemma:bd_stability}.
For any region $C\subseteq \Lambda$ define  a local version of the Hamiltonian $H_0$,
\be
H_0(C)=\sum_{\substack{A\in \calS(2)\\ A \subseteq C}} G_A.
\ee
Let $P_C$ and $Q_C$ be the local versions of the projectors $P$ and $Q$ defined in Eq.~(\ref{local_projectors}).
The following is a straightforward corollary of Lemma~\ref{lemma:bd_stability}.
\begin{corol}
\label{corol:bd_stability}
Let $C\subseteq \Lambda$ be any square of size smaller than $L^*$.
Let $W=\sum_{A\in \calS(q)} W_A$ be a perturbation satisfying Eqs.~(\ref{bd2},\ref{w}).
Suppose also that $W_A=0$ unless $C$ contains both $A$ and the nearest neighbors of $A$.
Then $W$ is relatively bounded by  $H_0(C)$ with a constant
$b=O(wq^2)$.
\end{corol}
\begin{proof}
Combining Eq.~(\ref{bd2}) and TQO-2 we conclude that $W_A P_C=0$ for all $A$.
Thus we can apply all steps in the proof of Lemma~\ref{lemma:bd_stability} to the square $C$
considered  as the entire lattice $\Lambda$.
\end{proof}

We shall need another technical lemma that provides a bound on the norm of a commutator involving
a block-diagonal Hamiltonian. Let us start from the simplest scenario.
Let $W$ be any operator such that $W$ is relatively bounded by $H_0$
with a constant  $0\le b<1$. Then for any operator $S$ we have
\bea
\| Q [S,W]P\| &=& \| Q W S P\| = \| Q W H_0^{-1} Q H_0 S P \|  \nn \\
&=&  \| Q W H_0^{-1} Q [H_0,S]P\| \le
 \| W H_0^{-1} Q\| \cdot \| Q [H_0,S]P\|.\nn
\eea
Here the first equality follows from $WP=0$
and the third equality uses $H_0P=0$.
Let $|\psi\ra\in Q$ be a normalized state such that $\| W H_0^{-1} Q\| =\| W H_0^{-1} \psi\|$.
Using the relative boundness assumption we get
\[
\| W H_0^{-1} Q\| =\| W H_0^{-1} \psi\| \le b\| H_0 H_0^{-1} \psi\|  = b.
\]
To conclude, we have proved that
\[
\| Q [S,W]P\| \le b\,   \| Q [S,H_0]P\|.
\]
%SBB1: simplified
Applying the same arguments as above with $P$ and $Q$ replaced by their local versions
$P_C$ and $Q_C$, see Eq.~(\ref{local_projectors}), we arrive at the following lemma.

%SBB1: simplified
\begin{lemma}
\label{lemma:commutator}
Let  $C\subseteq \Lambda$ be any region.
Let $W$ be any Hamiltonian such that $W$ is relatively bounded by
$H_0(C)$ with a constant $0\le b<1$.
Then for any operator $S$ one has
\be
\label{com2}
\| Q_C [S,W] P_C\| \le b\, \| Q_C [S,H_0(C)] P_C\|.
\ee
\end{lemma}
%SBB1: new stuff
Combining this lemma and  Corollary~\ref{corol:bd_stability}
we get a simple upper bound on $\| Q_C [S,W] P_C\|$.
\begin{corol}
\label{corol:commutator}
Let $C\subseteq \Lambda$ be any square of size smaller than $L^*$.
Let $W=\sum_{A\in \calS(q)} W_A$ be a perturbation satisfying Eqs.~(\ref{bd2},\ref{w}).
Suppose also that $W_A=0$ unless $C$ contains both $A$ and the nearest neighbors of $A$.
Then for any operator $S$ one has
\be
\label{com3}
\| Q_C [S,W] P_C\| \le O(wq^2) \, \| Q_C [S,H_0(C)] P_C\|.
\ee
\end{corol}

\section{Hamiltonian flow equations}
\label{sec:flow}

\subsection{Outline of the method}

We are interested in the low-energy spectrum of a perturbed Hamiltonian
$H_0+V$, where $V$ is a local perturbation specified by a list of local interactions,
\be
V=\sum_{r\ge 1} \sum_{A\in \calS(r)} V_{r,A}.
\ee
Here $V_{r,A}$ is some interaction supported on a square $A\in \calS(r)$.
Our strategy will be to reduce the case of a generic perturbation to the special
case of a block-diagonal perturbation which we can analyze using techniques of
Section~\ref{sec:rel}.  To this end we shall define a hierarchy of Hamiltonians unitarily equivalent to $H$,
\be
\label{H(n)}
H(n)=H_0 + \sum_{k=1}^n W(k) + V(n)+E(n) +\lambda(n) I, \quad n=0,1,2,\ldots,
\ee
such that $H(0)=H_0+V$ and $H(n+1)$ can be obtained from $H(n)$ by a unitary transformation,
\be
\label{step}
H(n+1)=U(n) H(n) U(n)^\dag, \quad U(n)U(n)^\dag=I.
\ee
Accordingly,  the spectrum of $H(n)$ is the same for all $n\ge 0$.
The purpose of the transformation $U(n)$ is to make the Hamiltonian
more close to the block-diagonal form.
We shall refer to $H(n)$
as a {\em level-$n$ Hamiltonian}.

Let us describe the purpose of various terms in Eq.~(\ref{H(n)}).
The Hamiltonian $W(k)$ represents a block-diagonal
contribution to the total Hamiltonian $H(n)$ that has been created at the level $k$.
The Hamiltonian $V(n)$ represents the part of the total Hamiltonian $H(n)$
that has to be block-diagonalized at the level $n+1$.
The Hamiltonians $V(n)$ and $W(n)$ will be represented by a sum of local interactions
supported in squares of size $r\le L^*$, and such that all interactions involved in $W(n)$
are individually block-diagonal,
\[
V(n)=\sum_{1\le r\le L^*} \sum_{A\in \calS(r)} V_{r,A}(n),
\]
and
\[
W(n)=\sum_{1\le r\le L^*} \sum_{A\in \calS(r)} W_{r,A}(n), \quad \mbox{where} \quad QW_{r,A}(n)P=0.
\]
All contributions from squares of size $r>L^*$ will be collected into the   third Hamiltonian $E(n)$
which can be regarded as an error Hamiltonian. The norm of $E(n)$ will be exponentially small in $L$
for all $n$.  The norms
of $W(n)$ and $V(n)$ will decay roughly as doubly-exponential functions of $n$.
Thus at level $n\sim \log{L}$ the total Hamiltonian $H(n)$ will be block-diagonal up to corrections of order $\exp{(-poly(L))}$
resulting from $V(n)$ and $E(n)$. Finally, $\lambda(n)$ is an overall energy shift that we shall often ignore.

We start at the level $n=0$ with initial conditions
\[
W_{r,A}(0)=0, \quad V_{r,A}(0)=V_{r,A} \quad \mbox{for $r\le L^*$}, \quad
E(0) = \sum_{r>L^*} \sum_{A\in \calS(r)} V_{r,A}.
\]
Accordingly, $H(0)=H_0+V$ is the Hamiltonian we are interested in.
Suppose we have already defined the Hamiltonians $W(0),\ldots,W(n)$, $V\equiv V(n)$, $E\equiv E(n)$
for some level $n$. Let
\[
W=\sum_{k=1}^n W(k)
\]
be the overall block-diagonal part  of $H(n)$.
Let
\[
H\equiv H(n) =H_0+W+V+E.
\]
We shall define the operator $U(n)$ in Eq.~(\ref{step}) as  $U(n)=\exp{(S)}$, where
$S$ is the solution of a linearized block-diagonalization problem
\be
\label{Sdef}
Q([S,H_0+W]+V)P=0, \quad S^\dag=-S.
\ee
The meaning of this equation can be easily understood if one treats $V$ as a perturbation and $H_0+W$ as
an unperturbed Hamiltonian. Expanding the transformed Hamiltonian $e^S He^{-S}$
in powers of $S$ we get
\[
e^S H e^{-S} = H_0 +W + ([S,H_0+W] + V) + O(S^2)+O(SV)+O(E).
\]
Thus Eq.~(\ref{Sdef}) says that the transformed Hamiltonian must be block-diagonal up to terms $O(V^2)$
and $O(E)$.
The solution of Eq.~(\ref{Sdef}) is constructed in Section~\ref{sec:lbd}, see Lemma~\ref{lemma:lbd}.
We start from defining raw versions of $W(n+1)$ and $V(n+1)$ which we shall
denote $\tilde{W}$ and $\tilde{V}$ respectively, namely
\be
\label{rawW}
\tilde{W}=[S,H_0+W]+V,
\ee
and
\be
\label{rawV}
\tilde{V}=e^S (H_0+W+V) e^{-S} - (H_0 + W +V   +[S,H_0+W]) .
\ee
A simple algebra shows that
\[
e^S H e^{-S}=
H_0+W+\tilde{W}+\tilde{V}+e^S E e^{-S}.
\]
Note also that $\tilde{W}$ is block-diagonal due to Eq.~(\ref{Sdef}).
We shall  construct a decomposition of $\tilde{W}$ into a sum
of local interactions that are individually block-diagonal using techniques of Section~\ref{sec:lbd}, see Corollary~\ref{corol:lbd}
of Lemma~\ref{lemma:lbd}. It will yield
\[
\tilde{W}=\sum_{r\ge 1} \sum_{A\in \calS(r)} \tilde{W}_{r,A}, \quad \mbox{where} \quad Q\tilde{W}_{r,A}P=0.
\]
 We shall construct a  decomposition of $\tilde{V}$ into a sum of local interactions
using Lemma~\ref{lemma:SV}  and  Lemma~\ref{lemma:second_order}  arriving at
\[
\tilde{V}=\sum_{r\ge 1} \sum_{A\in \calS(r)} \tilde{V}_{r,A}.
\]
Next we use $\tilde{W}$ and $\tilde{V}$ to define $W(n+1)$ and $V(n+1)$ by taking out all contributions from  squares
of size $r>L^*$ and adding these contributions  to the error Hamiltonian, that is,
\[
W(n+1)=\sum_{1\le r\le L^*} \sum_{A\in \calS(r)} \tilde{W}_{r,A},
\]
\[
V(n+1)=\sum_{1\le r\le L^*} \sum_{A\in \calS(r)} \tilde{V}_{r,A},
\]
and
\[
E(n+1)=e^SE e^{-S}+\sum_{r> L^*} \sum_{A\in \calS(r)} \tilde{W}_{r,A} + \tilde{V}_{r,A}.
\]
For the detailed analysis of these flow equations see Section~\ref{sec:theorem}

\subsection{Local decompositions of Hamiltonians}
\label{sec:dec}

In order to analyze convergence of the flow equations we shall need to set up some notations and
terminology.  Recall that $\calS(r)$ is a set of all $r\times r$ squares.
\begin{dfn}
Let $V$ be any operator acting on $\calH$. A local decomposition of $V$
is a list of operators  $\{V_{r,A}\}_{r,A}$ where $r\ge 1$ and
$A\in \calS(r)$ such that $V_{r,A}$ has support on a square $A$ and
\be
\label{ld}
V=\sum_{r\ge 1} \sum_{A\in \calS(r)} V_{r,A}
\ee
\end{dfn}
Note that a  local decomposition of an operator is not unique since the squares
involved in the decomposition overlap with each other. Nevertheless, we shall often identify
an operator and its local decomposition unless it may lead to confusions.
\begin{dfn}
A local decomposition of an operator $V$ is $(J,\mu,\alpha)$-decaying iff
\be
\label{decay}
\max_{r\ge 1} \max_{A\in \calS(r)} \| V_{r,A}\| \, r^\alpha e^{\mu r} \le J.
\ee
\end{dfn}
\noindent
Here we mean that $J,\mu,\alpha$ are some constants independent of $L$.
We shall often use a term $(J,\mu,\alpha)$-decaying operator  meaning that
this operator has a local decomposition which is $(J,\mu,\alpha)$-decaying.
To simplify notations we shall often use an abbreviation $(J,\mu)$-decay
for $(J,\mu,0)$-decay.

In the rest of this section we derive several auxiliary technical results
that can be skipped at the first reading.

We shall often need to construct a local decomposition for a commutator
$[S,V]$ given the local decompositions of $S$ and $V$.
%SBB1: the lemma moved to this section;
\begin{lemma}
\label{lemma:SV}
Suppose $S$ is $(K,\mu,\alpha)$-decaying and
$V$ is $(J,\mu,\beta)$-decaying for some $\alpha,\beta\ge 4$.
Then $[S,V]$ has a local decomposition which is $(cKJ,\mu)$-decaying for some
constant $c$.
\end{lemma}
\begin{proof}
Consider local decompositions of $S$ and $V$,
\be
S=\sum_{p\ge 1} \sum_{A\in \calS(p)} S_{p,A}, \quad \|S_{p,A}\| \le K p^{-\alpha} e^{-\mu p},
\ee
\be
V=\sum_{q\ge 1} \sum_{B\in \calS(q)} V_{q,B}, \quad \| V_{q,B}\| \le J q^{-\beta} e^{-\mu q}.
\ee
If $A\in \calS(p)$ and $B\in \calS(q)$ is a pair of non-overlapping squares, $A\cap B\ne\emptyset$,
then clearly $A\cup B$ can be covered by some square $C\in \calS(p+q)$.
Thus we can choose a local decomposition of $[S,V]$ as
\be
[S,V]=\sum_{r\ge 2} \sum_{C\in \calS(r)} D_{r,C},
\ee
where
\be
D_{r,C}=\sum_{p+q=r} \sum_{\substack{A\in \calS(p)\\A\subseteq C}} \; \;
\sum_{\substack{B\in \calS(q)\\ B\subseteq C}} \; \; [S_{p,A},V_{q,B}]\, \chi(A,B,C)
\ee
where $\chi(A,B,C)=0,1$ is some function that `distributes' the commutators $[S_{p,A},V_{q,B}]$
over different terms $D_{r,C}$. A specific form of this function is not important for us.
For fixed $p,q$ such that $p+q=r$ and a fixed $C\in \calS(r)$ we can bound the number of squares $A,B\subseteq C$ as
\[
\#\{ A\in \calS(p)\, : \, A\subseteq C\} =(r-p)^2=q^2
\]
and
\[
\#\{ B\in \calS(q)\, : \, B\subseteq C\}=(r-q)^2=p^2.
\]
It yields
\be
\|D_{r,C}\|\le 2KJe^{-\mu r} \sum_{p+q=r} p^{2-\alpha} q^{2-\beta} \le 2KJe^{-\mu r}\sum_{p,q\ge 1} p^{-2} q^{-2} = cKJe^{-\mu r}
\ee
for some constant $c$
provided that $\alpha,\beta\ge 4$.
\end{proof}

We shall also need the following simple lemma that will allow us to amplify the degree
$\alpha$ by any constant
by ``borrowing" some decay from the exponential function. It
makes the decay rate $\mu$ a bit smaller
and the amplitude $J$ a bit larger.
\begin{lemma}[\bf Degree Reset]
\label{lemma:reset}
Suppose $V$ is $(J,\mu,\beta)$-decaying. Let $0<\epsilon<1$ and $\alpha>0$ be any constants.
Then $V$ is also $(J',\mu',\alpha+\beta)$-decaying where
\be
J'=cJ^{1-\epsilon} \quad \mbox{and} \quad \mu'=\mu - J^{\frac{\epsilon}{\alpha}}.
\ee
Here $c$ is a constant depending on $\epsilon$ and $\alpha$.
\end{lemma}
\begin{proof}
Indeed, let $\mu'=\mu-\delta$ where $\delta$ will be chosen later. Then
\be
\| V_{r,A}\| \le J r^{-\beta} e^{-\mu r} \le J r^{-\alpha-\beta} e^{-\mu' r} \max_{q\ge 0} q^\alpha e^{-\delta q}\le
c\,  \delta^{-\alpha}  J r^{-\alpha-\beta} e^{-\mu' r},
\ee
where $c=\max_{x\ge 0} x^\alpha e^{-x}=O(1)$ is a constant. Choosing $\delta=J^{\epsilon/\alpha}$
we achieve the desired scaling.
\end{proof}
Finally, we shall use the following trivial observation.
\begin{lemma}
\label{lemma:boundary}
Suppose $V$ is $(J,\mu,\alpha)$-decaying. Then $V$
also has a local decomposition which is $(cJe^{2\mu},\mu,\alpha)$-decaying
with an extra property that $V_{r,A}$ acts trivially on all sites that are adjacent to
the boundary of $A$. Here $c=O(1)$ is some constant.
\end{lemma}
\begin{proof}
Indeed, replace each square $A\in \calS(r)$ in the decomposition of $V$
by a square $A'\in \calS(r+2)$ that contains $A$ and all nearest neighbors of $A$.
Let $V'_{r+2,A'}=V_{r,A}$. Then $V=\sum_{r\ge 1} \sum_{A\in \calS(r)} V'_{r,A}$
is the desired decomposition.
\end{proof}

\subsection{Proof of the  main theorem}
\label{sec:theorem}

In this section we prove Theorem~\ref{thm:main}.
\begin{proof}
We shall derive simplified flow equations for a triple  of parameters $J(n)$, $J_d(n)$, and $\mu(n)$
such that $V(n)$ is $(J(n),\mu(n),\beta)$-decaying and $W(n)$ is $(J_d(n),\mu(n),\alpha)$-decaying
for all $n\ge 0$. Here $\alpha$, $\beta$  are sufficiently large constants that will be chosen later.
Our manipulations with local decompositions will typically decrease the constants $\alpha$ and $\beta$,
so after each step of the flow equations we shall need to reset these constants back to their
original values using Lemma~\ref{lemma:reset}.

Since $W(0)=0$ we can choose initial conditions
\be
J(0)=J, \quad J_d(0)=0,  \quad \mbox{and} \quad \mu(0)=\mu.
\ee
Let us prove that for any constant $0<\epsilon<1$ there exist constants $c_1,c_2,c_3>0$  such that
for all $k\ge 0$ one has
\bea
J(k+1) &\le & c_1 J(k)^{2(1-\epsilon)}, \label{flow1} \\
J_d(k+1) &\le & c_2 J(k)^{1-\epsilon}, \label{flow2} \\
\mu(k+1) &= & \frac12 \mu(k) - c_3 J(k+1)^{\frac{\epsilon}{10}}, \label{flow3} \\
\|E(k+1)\| &\le & \|E(k)\| + O(L^3) J(k) e^{-c_3 L\mu(k) } \label{flow4}
\eea
provided that   $J(0)$ is below some constant threshold value.
Note that  although $\epsilon$ can be chosen arbitrarily close to $0$,
Eq.~(\ref{flow3}) does not permit one to choose $\epsilon=0$
since otherwise the decay rate $\mu(n)$ becomes negative after $O(1)$ iterations.

Supposed we have already proved Eqs.~(\ref{flow1}-\ref{flow4}) for
$k=0,1,\ldots,n$.
Let us denote  $V\equiv V(n)$, $E\equiv E(n)$,
\[
W\equiv \sum_{k=0}^n W(k)
\]
and
\[
J_d\equiv \sum_{k=0}^n J_d(k).
\]
Since $\mu(k)$ is monotone decreasing for $k=0,\ldots,n$, we can safely assume that
$W$ is $(J_d,\mu(n),\alpha)$-decaying. Also, since $J(k)$ decreases doubly exponentially
for $k=0,\ldots,n$ we can assume (for $k>0$) that
\be
J_d=O(J_d(1))=O(J^{1-\epsilon}).
\ee
Let $S$ be the solution of the linearized block-diagonalization
problem Eq.~(\ref{Sdef}) constructed in Lemma~\ref{lemma:lbd}.
By construction $S$ is $(K,\mu(n),\beta)$-decaying, where
\be
K=\frac{c_1J(n)}{1-c_2 J_d}=O(J(n)).
\ee
Here we assumed that $J$ is sufficiently small, so that $c_2 J_d \le 1/2$.
 Note that the assumptions of Lemma~\ref{lemma:lbd} require
$\beta\ge 2$ and $\alpha\ge \beta+4$.

Recall that $\tilde{W}$ and $\tilde{V}$ are raw versions of $W(n+1)$ and $V(n+1)$
defined in Eq.~(\ref{rawW}) and Eq.~(\ref{rawV}).
From Corollary~\ref{corol:lbd}, Section~\ref{sec:lbd1}, we infer that the operator $\tilde{W}$
has a local decomposition with block-diagonal terms which is $(cJ(n),\mu(n))$-decaying.
Note that the assumptions of Corollary~\ref{corol:lbd} require $\beta\ge 2$ and $\alpha\ge \beta+4$.
The operator $\tilde{V}$ can be rewritten as
\be
\tilde{V}=[S,V]+\omega(H),
\ee
where
$H\equiv H_0+W+V$ and $\omega(H)\equiv e^S H e^{-S} - H -[S,H]$.
We can assume that $H$ is $(c,\mu(n),\beta)$-decaying where $c=O(1)$
and we assumed that  $\beta\le \alpha$.
Applying Lemma~\ref{lemma:SV} from Section~\ref{sec:dec} we get a local decomposition for $[S,V]$
which is $(c_1 J(n)^2,\mu(n))$-decaying for some constant $c_1$ provided that $\beta\ge 4$.
Applying Lemma~\ref{lemma:second_order} from Section~\ref{sec:LR} we get a local decomposition for $\omega(H)$
which is $(c_2J(n)^2,\mu(n)/2,-1)$-decaying for some constant $c_2$ provided that $\beta\ge 6$.
Summarizing, $\tilde{V}$ is $(cJ(n)^2,\mu(n)/2,-1)$-decaying where $c$ is a constant.
It will be convenient to keep the decay rates of $\tilde{W}$ and $\tilde{V}$ the same.
Thus we shall assume that $\tilde{W}$ is $(cJ(n),\mu(n)/2)$-decaying (which is a weaker
version of what we proved above). One can easily check that a choice
\be
\alpha=10 \quad \mbox{and} \quad \beta=6
\ee
satisfies conditions of all lemmas used above.

Recall that $W(n+1)$ and $V(n+1)$ are defined by taking the local decompositions  of $\tilde{W}$ and $\tilde{V}$
and removing all terms associated with squares of size larger than $L^*$.
Therefore $W(n+1)$ and $V(n+1)$ have the same decay parameters as $\tilde{W}$ and $\tilde{V}$,
that is, we get
\be
J(n+1)\le c_1J(n)^2, \quad J_d(n+1) \le c_2J(n), \quad \mu(n+1)=\frac12 \mu(n)
\ee
for some constants $c_1,c_2$.
Note that  we have not reset $\alpha,\beta$ to their original values yet.
The total number of squares of size $r>L^*$ is at most $L^3$. Thus the contribution
to the error Hamiltonian $E(n+1)$ can be estimated as
\be
\|E(n+1)\| \le \|E(n)\| + O(L^3) J(n) e^{-\mu(n)L^*/2} = \|E(n)\| + O(L^3) J(n) e^{-c_3 \mu(n)L }
\ee
for some constant $c_3$. Resetting the constants $\alpha,\beta$ using Lemma~\ref{lemma:reset}
we arrive at the desired flow equations Eqs.~(\ref{flow1}-\ref{flow4}).

Solving Eq.~(\ref{flow1}) we get
\be
J(n)\le J\left(\frac{J}{J_0}\right)^{\theta^n}, \quad \theta =2(1-\epsilon)
\ee
for some constant $J_0>0$. Note that we are free to choose the constant $\epsilon$ as small as possible.
To simplify the formulas let us set $\theta=2$.
Also note that for small enough $J_0$ we can assume that $\mu(n)$ decays exponentially with
exponent arbitrarily close to $1/2$. To simplify the formulas let us assume that
\be
J(n)\le J\left(\frac{J}{J_0}\right)^{2^n} \quad \mbox{and} \quad \mu(n)=\mu 2^{-n}.
\ee
Simple algebra shows that choosing the number of steps $n=c\log{L}$ for some constant $c$
one can achieve the bounds
\be
\label{VEbounds}
\|V(n)\|, \|E(n)\| \le poly(L) \exp{(-c\sqrt{L})}.
\ee
Here the exponent $-c\sqrt{L}$ is determined by a tradeoff between the doubly exponential decay of
$J(n)$ and the exponential decay of $\mu(n)$.
Neglecting these exponentially small errors we can assume for simplicity that
the Hamiltonian $H(n)$ contains only block-diagonal contributions, that is,
\be
H(n)=H_0+W, \quad W=\sum_{k=0}^n W(k).
\ee
Here the local decomposition of $W$ is $(J_d,0,\alpha)$-decaying with $J_d=O(J^{1-\epsilon})$,
$\alpha=10$, and all terms in this decomposition are individually block-diagonal.
In addition, by definition of $W(k)$,  this decomposition contains only squares of size $r\le L^*$.
By performing an overall energy shift and using TQO-1 we can guarantee that
every local term in the decomposition of $W(k)$ has zero restriction on the $P$ subspace
(note that it increases the strength $J_d$ at most by a factor of two).
It allows us to  apply the machinery of Section~\ref{sec:rel}. In particular,
Lemma~\ref{lemma:bd_stability} says that
$W$ is relatively bounded by $H_0$ with a constant
\be
b\le O(J_d) \sum_{r\ge 1} r^{2-\alpha} = O(J_d).
\ee
Assuming that $b<1$, Lemma~\ref{lemma:relb}
implies that the spectrum of $H_0+W$ is contained in the union of intervals $I_k=(k(1-b),k(1+b))$,
where $k=0,1,2,\ldots$.
The effect of $V(n)$ and $E(n)$ can now be taken into account using the standard perturbation
theory by considering $H_0+W$ as an unperturbed Hamiltonian and using Eq.~(\ref{VEbounds}).

Finally, notice that since the dominant contribution to $W$ comes from the first level $k=1$,
we do not really need to perform the degree reset to estimate $b$ (since the degree reset
only changes our description of an operator but does not change the operator itself).
Hence we can set  $J_d=O(J)$.
The theorem is proved.
\end{proof}

\section{Linearized block-diagonalization problem}
\label{sec:lbd}

\subsection{Statement of the problem}

%SBB1: several bugs are fixed; new subsection on decomposition of [S,H_0+W]+V into local block-diagonal terms
Consider a pair of perturbations
\be
\label{VVV}
V=\sum_{1\le r\le L^*} \sum_{A\in \calS(r)} V_{r,A},
\ee
\be
\label{WWW}
W=\sum_{1\le r\le L^*} \sum_{A\in \calS(r)} W_{r,A},
\ee
such that all terms  $W_{r,A}$ are block-diagonal,
\be
\label{QWP=0}
QW_{r,A}P=0 \quad  \mbox{for all $r,A$}.
\ee

%SBB1:
A linearized block-diagonalization problem can be divided into two parts.
The first part is to find an anti-hermitian  operator $S$ such that
\be
\label{lin_bd}
Q( [S,H_0+W] + V )P=0, \quad S^\dag=-S
\ee
and construct a local decomposition of $S$.
The second part is to construct a local decomposition for the transformed Hamiltonian
\be
\tilde{W}=[S,H_0+W] +V=\sum_{r\ge 1} \sum_{A\in \calS(r)} \tilde{W}_{r,A}
\ee
such that every term in the local decomposition is block-diagonal, that is, $Q\tilde{W}_{r,A}P=0$
 for all $r,A$. We solve the two parts of the problem
in Lemma~\ref{lemma:lbd} and its Corollary~\ref{corol:lbd} respectively.
Throughout this section we assume that $H_0$ is a Hamiltonian defined in Eq.~(\ref{H_0})
that obeys conditions TQO-1 and TQO-2.

\subsection{Finding the transformation $S$}
In this section we prove the following
%SBB1: we need a stronger version of the lemma, otherwise we have problems with local decomposition of [S,H_0+W]+V
%into local terms that are individually block diagonal. The proof is exactly the same as before with some obvious modifications
%near the end of the proof
\begin{lemma}
\label{lemma:lbd}
Let $V$ and $W$ be perturbations defined in Eqs.~(\ref{VVV},\ref{WWW},\ref{QWP=0}).
Suppose that
$V$ is $(J,\mu,\beta)$-decaying and $W$ is $(J_d,\mu,\alpha)$-decaying
such that  $\beta\ge 2$ and $\alpha\ge \beta+4$.
Then
there exist  constants $c_1,c_2>0$ such that Eq.~(\ref{lin_bd}) has
a solution $S$ which is $(K,\mu,\beta)$-decaying with
\be
K= \frac{c_1J}{1-c_2J_d},
\ee
\end{lemma}
\begin{proof}
%SBB1: new paragraph
Let us start from several simplifying assumptions.
Without loss of generality we can assume that
\be
\label{WP=0}
W_{r,A} P=0  \quad \mbox{for all $r,A$}.
\ee
Indeed,  condition TQO-1 guarantees that $W_{r,A} P$ is  a multiple  of $P$.
Shifting $W_{r,A}$ by the corresponding multiple of the identity
we can satisfy Eq.~(\ref{WP=0}). On the other hand, one can easily check that
Eq.~(\ref{lin_bd}) in invariant under such a shift. Note also that the shift can increase the norm
$\|W_{r,A}\|$  at most by a factor of two. Indeed, if
$W_{r,A}P=cP$ then $|c|\le \|W_{r,A}\|$ and thus $\|W_{r,A}- cI\|\le \|W_{r,A}\|+|c| \le 2\|W_{r,A}\|$.
Thus we can assume that $W$ satisfies Eq.~(\ref{WP=0}) provided that
we change $J_d$ to $2J_d$ in the final answer.
By the same reasons, we can assume that
\be
\label{PVP=0}
PV_{r,A}P=0   \quad \mbox{for all $r,A$}
\ee
provided that we change $J$ to $2J$ in the final answer.
In addition, we can assume that $V_{r,A}$ commutes with $G_B$, $B\in \calS(2)$, whenever
$B$ is not contained in $A$,
\be
\label{boundary}
[V_{r,A},G_B]=0 \quad \mbox{for all $B\in \calS(2)$ such that $B\cap A^c\ne \emptyset$}.
\ee
Indeed, by adding an idle layer of sites to each square in the decomposition of $V$ as explained in Lemma~\ref{lemma:boundary} we guarantee Eq.~(\ref{boundary}). The price we pay for this simplification is that $J$ has to be changed to
$cJe^{2\mu}=O(J)$ in the final answer, see Lemma~\ref{lemma:boundary}.
Note also that now the local decomposition of $V$ in Eq.~(\ref{VVV}) starts from squares of size $r=3$,
\be
V=\sum_{3\le r\le L^*} \sum_{A\in \calS(r)} V_{r,A}.
\ee

We shall construct a solution $S$ as a series
$S=\sum_{i=1}^\infty S^{(i)}$, where $S^{(i)}$ is anti-hermitian for all $i$ and
\bea
Q([S^{(1)},H_0]+V)P&=& 0, \label{S1} \\
Q([S^{(i+1)},H_0] + [S^{(i)},W])P &=& 0, \quad \mbox{for $i\ge 1$}. \label{Si}
\eea
For any region $A\subseteq \Lambda$ define a restricted Hamiltonian
\be
\label{localH_0}
H_0(A) = \sum_{\substack{B\in \calS(2)\\ B\subseteq A}} G_B.
\ee
Define also a super-operator $\calE_A$ that takes as argument an arbitrary operator $O$
and returns an operator
\be
\calE_A(O) = Q_A H_0(A)^{-1} O P_A - P_A O H_0(A)^{-1} Q_A.
\ee
Note that the $P_A$ is the zero-subspace of $H_0(A)$,  so that $Q_A H_0(A)^{-1}$ is well-defined.
\begin{prop}
\label{prop:EA}
Let $O_A$ be any operator acting on $A$ such that $PO_AP=0$.
Suppose $O_A$ commutes with $G_B$, $B\in \calS(2)$,
whenever $B$ is not contained in $A$. Then
\be
\label{off_diag}
Q\left( [\calE_A(O_A),H_0] + O_A  \right)P = 0.
\ee
If $O_A$ is hermitian then $\calE_A(O_A)$ is anti-hermitian.
\end{prop}
\begin{proof}
Indeed, since  all terms in $H_0$ which are not contained in $A$
commute with $\calE_A(O_A)$ while $H_0(A) P_A=0$ we have
\[
[\calE_A(O_A),H_0]=-Q_A O_A P_A - P_A O_A Q_A.
\]
It yields
\[
Q[\calE_A(O_A),H_0]P = -Q Q_A O_A P = -Q( Q_A+P_A) O_A P = -QO_A P,
\]
where the first equality follows from  $P_A P=P$, $Q_A P =0$, and the second equality
uses identity $P_A O_A P = PO_A P =0$. Thus we have proved Eq.~(\ref{off_diag}).
The last statement of the proposition is obvious.
\end{proof}
Using the assumptions Eqs.~(\ref{PVP=0},\ref{boundary}) and the proposition
we can choose $S^{(1)}$ in Eq.~(\ref{S1}) as
\be
\label{S1choice}
S^{(1)}=\sum_{r\ge 3} \sum_{A\in \calS(r)} S^{(1)}_{r,A}, \quad S^{(1)}_{r,A}= \calE_A(V_{r,A}).
\ee
Taking into account that
\be
\| \calE_A(O_A)\| \le \|O_A\| \quad \mbox{for any $O_A$}
\ee
we conclude that Eq.~(\ref{S1choice})
is a local decomposition of $S^{(1)}$ which is $(K_1,\mu,\beta)$-decaying where
\be
\label{K1}
K_1=J.
\ee
This decomposition has an extra property that $S^{(1)}_{r,A}$
is block-off-diagonal with respect to $P_A$, $Q_A$, and $S^{(1)}_{r,A}$
commutes with  $G_B$, $B\in \calS(2)$, whenever  $B$ is not contained in $A$.

Let us now solve Eq.~(\ref{Si}).
We shall assume as our induction hypothesis that $S^{(i)}$ possesses a local decomposition
\be
\label{ind1}
S^{(i)}=\sum_{p\ge 3} \sum_{A\in \calS(p)} S^{(i)}_{p,A}
\ee
such that
\begin{enumerate}
\item[\bf I1] $S^{(i)}_{p,A}$ is block-off-diagonal with respect to $P_A$, $Q_A$
\item[\bf I2]  $S^{(i)}_{p,A}$ commutes with $G_B$, $B\in \calS(2)$, whenever $B$
 is not contained in $A$
\item[\bf I3] $\| \, [H_0(A),S^{(i)}_{p,A}] \, \| \le K_i\,  p^{-\beta}  e^{-\mu p}$
\end{enumerate}
Here the coefficient $K_i$ will be determined inductively in terms of $K_{i-1}$.
%SBB1: added more explanations
Note that combining (I1), (I3) with the fact that $Q_A H_0(A)\ge I$ one gets
\[
\| S^{(i)}_{p,A} \| = \| Q_A  S^{(i)}_{p,A} P_A \| \le \| Q_A H_0(A)S^{(i)}_{p,A} P_A \| = \| \, [H_0(A),S^{(i)}_{p,A}] \, \| \le K_i\, p^{-\beta}  e^{-\mu p},
\]
that is, $S^{(i)}$ is $(K_i,\mu,\beta)$-decaying.

The base of induction is $i=1$ which we have already proved.
Our first step will be choosing a local decomposition for  the commutator $[S^{(i)},W]$
in Eq.~(\ref{Si}).
We shall need the following geometrical fact.
%SBB1: if p+q is larger than L we need to choose C=entire lattice;
%Also the proposition is wrong for p=1 so I excluded this case
\begin{prop}
\label{prop:geom}
Let $A\in \calS(p)$, $p\ge 2$, and $B\in \calS(q)$, $q<L$, be any squares such that $A\cap B\ne \emptyset$.
Then there exists a square $C\in \calS(r)$, $r=\min{(p+q,L)}$, such that
$A\cup B \subseteq  C$ and $C$ contains all nearest neighbors of $B$.
\end{prop}
\begin{proof}
If $r=L$ the statement is obvious, so assume $r=p+q<L$.
Define a metric on the lattice using the $l_\infty$-norm, that is,
if $u=(u_x,u_y)$ and $v=(v_x,v_y)$ is a pair of sites, then
\[
D(u,v)=\max{\{ |u_x-v_x|,|u_y-v_y|\}}.
\]
For any region $M\subseteq \Lambda$ let $D(M)$ be the diameter of $M$, i.e., the
largest distance between a pair of sites $u,v\in M$. Clearly $D(A)=p-1$ and $D(B)=q-1$.
Since $A\cap B\ne \emptyset$ we have $D(A\cup B)\le D(A)+D(B)=p+q-2$.
Therefore, $A\cup B$ can be covered by a square $C'\in \calS(p+q-1)$.
If $C'=B$ one actually has $C'\in \calS(q)$. Now one can choose $C$
as an arbitrary square of size $r$ that contains $C'$ and all nearest neighbors of $C'$.
Suppose now that $C'\ne B$. Then  either $C'$ contains all nearest neighbors of $B$, or
$C'$  shares an edge or a corner with $B$. In the latter case, we can extend the size of $C'$
by one obtaining a square $C\in \calS(p+q)$ with the desired properties.
\end{proof}
The proposition implies that
\be
\label{Si_com}
[S^{(i)},W]=\sum_{r\ge 3} \sum_{C\in \calS(r)} D^{(i)}_{r,C},
\ee
where
\be
\label{Di}
D^{(i)}_{r,C}=\sum_{p+q=r} \; \; \sum_{\substack{A\in \calS(p)\\ A \subseteq C}}\;
\sum_{\substack{B\in \calS(q)\\ B\subseteq C}}   [S^{(i)}_{p,A},W_{q,B}]\, \chi(A,B,C)
\ee
and $\chi(A,B,C)=0,1$ is some function that `distributes' the commutators over
different terms $D^{(i)}_{r,C}$. A particular choice of such distribution is not important for us.
Using Proposition~\ref{prop:geom} we can assume that $\chi(A,B,C)=0$ unless
$A\cup B\subseteq C$ and $C$ contains all nearest neighbors of $B$.
By construction, $D^{(i)}_{r,C}$ has support on $C$,
that is, Eqs.~(\ref{Si_com},\ref{Di}) define a local decomposition of the commutator $[S^{(i)},W]$.

Using Proposition~\ref{prop:EA} we can choose a solution $S^{(i+1)}$ of Eq.~(\ref{Si}) as
\be
\label{Si1}
S^{(i+1)} = \sum_{r\ge 3} \sum_{C\in \calS(r)} S^{(i+1)}_{r,C}, \quad  S^{(i+1)}_{r,C}=\calE_C(D^{(i)}_{r,C}).
\ee
This is a local decomposition of $S^{(i+1)}$ that satisfies (I1) by definition of the map $\calE_C$. Let us check that it satisfies (I2). Indeed, let $G_F$, $F\in \calS(2)$, be such that $F$ is not contained in $C$.
Then $F$ is not contained in $A\subseteq C$ and thus $G_F$ commutes with all $S^{(i)}_{p,A}$
in Eq.~(\ref{Di}). Since for all terms $W_{q,B}$ in Eq.~(\ref{Di}) the square
$C$ contains both $B$ and the nearest neighbors of $B$, we conclude that $F$ does not overlap with $B$,
that is, $G_F$ commutes with $W_{q,B}$. Finally, $G_F$ commutes with all Krauss operators
involved in the map $\calE_C$. Thus $G_F$ commutes with $S^{(i+1)}_{r,C}$ which proves (I2).

It remains to verify that the local decomposition Eq.~(\ref{Si1}) satisfies (I3).
It will be convenient to introduce an auxiliary Hamiltonian
\be
W_q(A,C) = \sum_{\substack{B\in \calS(q)\\ B\subseteq C}} \chi(A,B,C) \, W_{q,B}.
\ee
It allows us to rewrite Eq.~(\ref{Di}) as
\be
D^{(i)}_{r,C}=\sum_{p+q=r} \; \;\sum_{\substack{A\in \calS(p)\\ A\subseteq C}}\;   [S^{(i)}_{p,A}, W_q(A,C)].
\ee
%SBB1: simplified
Applying  Corollary~\ref{corol:commutator}, using Eq.~(\ref{WP=0}),  and taking into account that
$W$ is $(J_d,\mu,\alpha)$-decaying we get
\bea
\label{I3proof1}
\| \, [ H_0(C), S^{(i+1)}_{r,C}]\, \| &=&
\| Q_C D^{(i)}_{r,C} P_C  \|  \le
\sum_{p+q=r} \; \;\sum_{\substack{A\in \calS(p)\\ A\subseteq C}}\;
\| Q_C [S^{(i)}_{p,A}, W_q(A,C)]  P_C \| \nn \\
&\le &
\sum_{p+q=r} \; \;\sum_{\substack{A\in \calS(p)\\ A\subseteq C}}\;
b_q \, \| Q_C [H_0(C), S^{(i)}_{p,A}] P_C \|,
\eea
where
\be
b_q \le c J_d q^{2-\alpha} e^{-\mu q}
\ee
for some constant $c$.
From (I1) we infer that $S^{(i)}_{p,A}$ is block-off-diagonal with respect to $P_A$, $Q_A$
while  (I2) implies $[H_0(C), S^{(i)}_{p,A}]=[H_0(A), S^{(i)}_{p,A}]$.
Taking into account that  $P_C=P_AP_C$ we get a bound
\be
\label{I3proof2}
 \| Q_C [H_0(C), S^{(i)}_{p,A}] P_C \| \le
  \| Q_A [H_0(A), S^{(i)}_{p,A}] P_A\| =   \| \, [H_0(A), S^{(i)}_{p,A}]\, \|  \le K_i \, p^{-\beta}  e^{-\mu p},
 \ee
where the last inequality follows from (I3).
Combining Eqs.~(\ref{I3proof1},\ref{I3proof2}) and noting that
the number of squares $A\in \calS(p)$ such that $A\subseteq C$ is equal to
$(r-p)^2=q^2$ we get
\bea
\| \, [ H_0(C), S^{(i+1)}_{r,C}]\, \| &\le & K_i \sum_{q=1}^{r-1} q^2 b_q\, (r-q)^{-\beta}  e^{-\mu (r-q)}   \nn \\
 &\le &  cJ_d K_i \, e^{-\mu r} \sum_{q=1}^{r-1}  q^{4-\alpha}(r-q)^{-\beta}.
\eea
It is convenient to split the sum over $q$ into two parts:
\be
\sum_{1\le q\le r/2}  q^{4-\alpha}(r-q)^{-\beta} \le cr^{-\beta} \sum_{q\ge 1} q^{4-\alpha} =c'r^{-\beta}
\ee
for some constants $c,c'$ since we assumed $\alpha\ge \beta+4\ge 6$.
As for the other part, we have
\be
\sum_{r/2\le q\le r-1}  q^{4-\alpha}(r-q)^{-\beta} \le cr^{4-\alpha} \sum_{q\ge 1} q^{-\beta} =c'r^{4-\alpha} \le c'r^{-\beta}
\ee
for some constants $c,c'$
since we assumed that $\beta\ge 2$ and $\alpha\ge \beta+4$.
Therefore we arrive at
\be
\| \, [ H_0(C), S^{(i+1)}_{r,C}]\, \|\le  cJ_d K_i \, r^{-\beta}  e^{-\mu r}
\ee
for some constant $c$ which
 proves (I3) for
\be
\label{Kinduction}
K_{i+1}=cJ_d K_i.
\ee
Thus all induction assumptions are proved for $S^{(i+1)}$.
By obvious reasons  $S=\sum_{i\ge 1} S_i$ is $(K,\mu,\beta)$-decaying with
\be
K\le \sum_{i\ge 1} K_i = \frac{c_1J}{1-c_2J_d}.
\ee
\end{proof}

\subsection{Local decomposition of the transformed Hamiltonian}
\label{sec:lbd1}
%SBB1: new section
Lemma~\ref{lemma:lbd} has the following corollary.
\begin{corol}
\label{corol:lbd}
Let $V$, $W$, and $S$ be as in Lemma~\ref{lemma:lbd}.
Suppose $V$ is $(J,\mu,\beta)$-decaying and
$W$ is $(J_d,\mu,\alpha)$-decaying for some
$\beta\ge 2$ and $\alpha\ge \beta+4$.
Then a transformed Hamiltonian $\tilde{W}=[S,H_0+W] +V$ has a
 local decomposition
\be
\tilde{W}=\sum_{r\ge 1} \sum_{A\in \calS(r)} \tilde{W}_{r,A}
\ee
such that  $Q\tilde{W}_{r,A}P=0$ for all $r,A$.
This decomposition is $(\tilde{J}_d,\mu)$-decaying, where
\be
\tilde{J}_d \le \frac{cJ}{1-cJ_d}
\ee
for some constant $c$.
\end{corol}
\begin{proof}
We shall use notations and techniques introduced in the proof of Lemma~\ref{lemma:lbd}.
By definition of $S$ we have
\be
\label{tildeWdecom}
\tilde{W}=\sum_{i\ge 1} W^{(i)},
\ee
where
\be
W^{(1)} =[S^{(1)},H_0]+V, \quad \mbox{and} \quad  W^{(i)}=[S^{(i)},H_0]+[S^{(i-1)},W] \quad \mbox{for $i\ge 2$}.
\ee
Let us choose the local decomposition of $W^{(1)}$ as
\be
W^{(1)}=\sum_{r\ge 3} \sum_{A\in \calS(r)} W^{(1)}_{r,A},
\ee
where
\be
W^{(1)}_{r,A}=[\calE_A(V_{r,A}),H_0] + V_{r,A}=P_A V_{r,A} P_A + Q_A V_{r,A} Q_A.
\ee
Here the last equality uses Proposition~\ref{prop:EA} (see the first equation in the proof of the proposition).
Obviously, $\|W^{(1)}_{r,A}\| \le \|V_{r,A}\|$ and thus $W^{(1)}$ is $(J,\mu,\beta)$-decaying.
As was noticed in the proof of Lemma~\ref{lemma:lbd} we can assume that $V_{r,A}$
commutes with all operators $G_B$, $B\in \calS(2)$ for which $B$ is not contained in $A$.
It means that
\be
P W^{(1)}_{r,A} = P V_{r,A} P_A = PV_{r,A}P = W^{(1)}_{r,A} P,
\ee
that is, $W^{(1)}_{r,A}$ is block-diagonal.

Recall that we have a decomposition
\be
[S^{(i)},W]=\sum_{r\ge 3} \sum_{C\in \calS(r)} D^{(i)}_{r,C},
\ee
where $D^{(i)}_{r,C}$  commutes with
all operators $G_B$, $B\in \calS(2)$ for which $B$ is not contained in $C$, see
Eqs.~(\ref{Si_com},\ref{Di}) in the proof of Lemma~\ref{lemma:lbd}.
It means that we can choose the local decomposition of $W^{(i)}$ with $i\ge 2$ as
\be
W^{(i)}=\sum_{r\ge 3} \sum_{C\in \calS(r)} W^{(i)}_{r,C},
\ee
where
\be
W^{(i)}_{r,C}=[\calE_C(D^{(i-1)}_{r,C}),H_0]+ D^{(i-1)}_{r,C} = P_C D^{(i-1)}_{r,C} P_C + Q_C D^{(i-1)}_{r,C} Q_C,
\ee
see Eq.~(\ref{Si1}).
Obviously, $\| W^{(i)}_{r,C}\| \le \| D^{(i-1)}_{r,C}\|$ and $W^{(i)}_{r,C}$ is block-diagonal.
Using the fact that $S^{(i)}$ is $(K_i,\mu,\beta)$-decaying where $K_i$ is defined
by Eqs.~(\ref{K1},\ref{Kinduction}), and using Lemma~\ref{lemma:SV}
we conclude that $W^{(i)}$ is $(cJ_d K_{i-1},\mu)$-decaying.
Using Eq.~(\ref{tildeWdecom}) we obtain a local decomposition
of $\tilde{W}$ with individually block-diagonal local terms which is $(\tilde{J}_d,\mu)$-decaying
with
\be
\tilde{J}_d=J + \sum_{i\ge 2} cJ_d K_{i-1} \le J + \sum_{i\ge 1} J(cJ_d)^i = \frac{J}{1-cJ_d}.
\ee
Finally we have to replace $J$ and $J_d$ by $O(J)$ and $O(J_d)$ to justify our
assumptions Eqs.~(\ref{WP=0},\ref{PVP=0},\ref{boundary}).
\end{proof}

\section{Lieb-Robinson bounds}
\label{sec:LR}

Let $S$ be some anti-hermitian operator and $V$ be an arbitrary operator.
We shall use notations
\bea
\tau(V)&=&e^S V e^{-S}, \nn \\
\omega(V)&=& \tau(V)-V-[S,V].
\eea
Suppose we are given some local decompositions of $S$ and $V$.
In order to get a closed system of flow equations we need to construct a local decomposition for
$\omega(V)$, see Section~\ref{sec:flow}.
The main result of this section is the following lemma.
%SBB1: \beta \ge 6
\begin{lemma}
\label{lemma:second_order}
Suppose $S$ is $(K,\mu,\alpha)$-decaying for some $\alpha\ge 6$.
Suppose $V$ is $(J,\mu,\beta)$-decaying for some $\beta\ge 6$.
Then $\omega(V)$ has a local decomposition which is $(cJ K^2, \mu/2,-1)$-decaying
for some constant $c$.
\end{lemma}
Thus the magnitude of interactions of range $r$ in $\omega(V)$ decays as $cJ K^2 re^{-\mu r/2}$.
Our arguments will rely on the Lieb-Robinson bound, see~\cite{hastings-koma}. More specifically,
we shall exploit  quasi-locality of
dynamics in quantum spin systems with a fast decay of interactions as presented in~\cite{NS07}.
The extra factor $1/2$ in the decay rate of $\omega(V)$
represents a simple geometrical fact that the diameter of the light-cone of any local region
increases with a rate $2v$, where $v$ is the Lieb-Robinson velocity.
This extra factor $1/2$  is the price one has to pay for using the powerful machinery
built on the Lieb-Robinson bound.

We shall start from solving a  somewhat simpler problem.
Let $O$ be any operator with support  on some square $B\in \calS(q)$.
Consider a local decomposition of $S$,
\be
S=\sum_{p\ge 2} \sum_{A\in \calS(p)} S_{p,A}.
\ee
Let $C\in \calS(q+2j)$ be a square that contains $B$ and all sites within
distance $j$ from $B$ (with respect to the $l_\infty$-distance), that is,
\be
C=\{ u\in \Lambda\, : \, D(u,B)\le j\}.
\ee
Let $S_C$ be a localized version of $S$ obtained by taking out all interactions
whose support is not contained in $C$,
\be
S_C= \sum_{p\ge 2} \; \sum_{\substack{A\in \calS(p) \\ A\subseteq C}} S_{p,A}.
\ee
Define also a localized version of $\omega(O)$, that is,
\be
\omega_C(O)=e^{S_C} O e^{-S_C} - O - [S_C,O].
\ee
By definition, $\omega_C(O)$ has support on $C$.
We shall need a  bound on the difference $\|\omega(O)-\omega_C(O)\|$
that is proportional to $\|O\| \cdot K^2e^{-\mu j}$.
\begin{lemma}[\bf Quasi-Local Dynamics]
\label{lemma:QLD}
Let $S$, $S_C$ and $O$ be the operators defined above.
Suppose $S$ is $(K,\mu,\alpha)$-decaying for some $\alpha \ge 6$.
 Then there exist constants $c_0,c_1>0$ such that
\be
\label{LR}
\|  \omega_C(O) - \omega(O) \|  \le c_0 (q+j)q^4  K^2 \| O\|\,e^{-\mu j}
\ee
whenever $K\le c_1$.
\end{lemma}
%SM: The proof is now complete. Note the extra q^2 (q+2j) [1+(q+2j) j^{-\alpha}] factor and the explicit e^{8 C_0(1) K}, which we bound by assuming $K \le c_1$. I think it is important to get rid of the factor (q+2j) if we can since it propagates to Lemma 9 as a linear growth. The factor (q+2j) j^{-\alpha} is fine because of the j^{-\alpha} term.
\begin{proof}
Define
\be
\tau^t(O)=e^{St} O e^{-St}, \quad \tau^t_C(O)=e^{S_C t} O e^{-S_C t}.
\ee
For any $0\le t\le 1$ define an operator
\be
f(t)=\tau_C^t(O) -\tau^t(O) - [S_C-S,O]t.
\ee
Note that $f(t)$ is an analytic function  and
\be
\label{in_con}
f(0)=\dot{f}(0)=0, \quad f(1) = \omega_C(O)-\omega(O).
\ee
Computing the derivatives over $t$ we get
\be
\dot{f}(t)=[S_C,\tau_C^t(O)] - [S,\tau^t(O)] - [S_C-S,O],
\ee
\be
\ddot{f}(t)=[S_C,[S_C,\tau_C^t(O)]] - [S,[S,\tau^t(O)]].
\ee
Let us extract from $\ddot{f}(t)$ a norm preserving term $[S,\dot{f}(t)]$. After simple
algebra we get
\be
\ddot{f}(t)=[S,\dot{f}(t)] - [S-S_C,[S_C,\tau_C^t(O)]] - [S,[S-S_C,O]].
\ee
It means that
\be
\dot{f}(t) = \tau^t \left(\int_0^t \tau^{-s} \left([[S_C,\tau_C^s(O)],\Delta S_C] + [[\Delta S_C,O],S]\right) ds\right),
\ee
where $\Delta S_C\equiv S-S_C$.
Taking into account the initial conditions Eq.~(\ref{in_con}) we get
\be
\label{|f(1)|}
\|f(1)\|\le \int_0^1 dt_1 \, \|\dot{f}(t_1)\| \le
\frac12  \|[[\Delta S_C,O], S] \| +  \int_0^1 dt_1 \int_0^{t_1} dt_2\, \| [\tau_C^{t_2}([S_C,O]),\Delta S_C] \|.
\ee
Let us start from bounding the time-independent term  $\| [[\Delta S_C,O],S] \|$.
Denote
\[
\Gamma_{p_2,p_1}=\# \{(E_2,E_1)\, :\,  E_2 \in \calS(p_2), \quad E_1\in \calS(p_1), \quad  E_2\cap (E_1\cup B)\ne \emptyset, \quad E_1\cap B \ne \emptyset \}.
\]
One can easily check that
\[
\Gamma_{p_2,p_1} \le  (q+p_1+p_2)^2 (q+p_1)^2\le 2 (q+p_1)^4+2(q+p_1)^2  p_2^2.
\]
The commutator $[\Delta S_C,O]$ has contributions only from squares of size $\ge j$ in the decomposition of $S$, which
implies
\[
\| [[\Delta S_C,O],S] \| \le cK^2 \| O\| \sum_{p_1\ge j} \; \; \sum_{p_2\ge 1} \; \; \Gamma_{p_2,p_1}\,  p_1^{-\alpha} p_2^{-\alpha} e^{-\mu(
p_1+p_2)}.
\]
Here and below $c$ stands for a constant factor.
Since $\alpha \ge 6$ the sum over $p_2$ is bounded by a constant and we arrive at
\be
\label{dbcm1}
\| [[\Delta S_C,O],S] \| \le cK^2 \| O\| e^{-\mu j} \sum_{p_1\ge j} (q+p_1)^4 p_1^{-\alpha} = cq^4 K^2 \| O\| e^{-\mu j}.
\ee

Let us now  bound  $\| [\tau_C^{t_2}([S_C,O]),\Delta S_C] \|$. Note that the commutator has
contributions only from those terms in the decomposition of $\Delta S_C$ that overlap with both $C$
and its complement $C^c$. Define
\[
\Omega_{p_1,p_2} (t)= \max_{\substack{E_1\in \calS(p_1)\\ E_1\cap B \ne \emptyset}}\; \;
\max_{\substack{E_2\in \calS(p_2)\\ E_2\cap C\ne \emptyset \\E_2\cap C^c\ne \emptyset }} \; \;
\max_{O_1,O_2}\;\;
\| \, [\tau^t_C([O_1,O]),O_2] \,\|,
\]
where $O_1,O_2$ are unit-norm operators acting on $E_1,E_2$ respectively.
Counting the number of squares contributing to the double commutator yields
\[
\| \, [\tau_C^{t}([S_C,O]),\Delta S_C] \, \| \le c(q+j)K^2 \sum_{p_1,p_2\ge 1} \; \;
(q+p_1)^2 p_1^{-\alpha} p_2^{2-\alpha} e^{-\mu(p_1+p_2)} \, \Omega_{p_1,p_2}(t).
\]
As we show below, the unitary evolution under $S_C$ can be characterized by
a finite Lieb-Robinson velocity $v_{LR}=O(K)$.  Using the Lieb-Robinson bound from~\cite{NS07}
that governs unitary evolution under Hamiltonians with exponentially decaying interactions
one gets
\[
\Omega_{p_1,p_2}(t) \le c\| O\| (p_1+q)^2 p_2^2 \exp{\left[ v_{LR} t - \mu \theta(j-p_1-p_2)\right] },
\]
where $\theta(x)=x$ for $x\ge 0$ and $\theta(x)=0$ for $x<0$. Note that $\theta(j-p_1-p_2)$
is a lower bound on the distance between supports of $[O_1,O]$ and $O_2$ in the definition of $\Omega_{p_1,p_2}$.
Clearly, $p_1+p_2+\theta(j-p_1-p_2) \ge j$. Since $0\le t\le 1$ we can assume that $v_{LR} t =O(Kt)=O(1)$. Hence
\be
\label{dbcm2}
\| \, [\tau_C^{t}([S_C,O]),\Delta S_C] \, \| \le c(q+j)\|O\| K^2 e^{-\mu j} \sum_{p_1,p_2\ge 1} \; \;
(q+p_1)^4 p_1^{-\alpha} p_2^{4-\alpha}\le c(q+j)q^4 \|O\| K^2 e^{-\mu j}
\ee
for $\alpha\ge 6$.
Combining Eqs.~(\ref{dbcm1},\ref{dbcm2}) and computing the integrals in Eq.~(\ref{|f(1)|}) we arrive at
\[
|f(1)| \le c(q+j)q^4 \|O\| K^2 e^{-\mu j}.
\]

It remains to check that we have fast enough decay of interactions in $S$ to use the Lieb-Robinson bound from Ref.~\cite{NS07}.
 Below, we use notations from Ref.~\cite{NS07}.
Define a function
\be
F_\mu(x)=\frac{\exp{(-\mu x)}}{1+x^2}.
\ee
As was shown in Ref.~\cite{NS07} the Lieb-Robinson velocity is bounded by a multiple of:
\be
\| S\|_\mu:= \sup_{u,v\in \Lambda} \sum_{r\ge 1} \sum_{\substack{A\in \calS(r)\\ A\ni u,v}} \frac{\| S_{r,A}\|}{F_\mu(D(u,v))},
\ee
with the multiplying factor being $2 C_0(1)$, where $C_0(1)$ is a numerical constant (depending only on the dimensionality of the system; in our case the dimension is $2$.)
For any pair of sites $u,v\in A$ with $A\in \calS(r)$ one has $D(u,v)\le r$ (here and below we use
the $l_\infty$-distance). Since $F_\mu(x)$ is monotone-decreasing we have $F_\mu(D(u,v))\ge F_\mu(r)$.
Taking into account that the number of squares $A\in \calS(r)$ such that $A\ni u$ is at most $r^2$ we get
\be
\| S\|_\mu\le K\sum_{r\ge 1} r^2 (1+r^2) r^{-\alpha} \le c K, \quad c= \sum_{r\ge 1} (r^{-2}+r^{-4}) \le 4,
\ee
provided that $\alpha\ge 6$. Thus, the Lieb-Robinson velocity is bounded by $8 C_0(1) K$.
\end{proof}
%% End of Quasi-Local Dynamics %%
Let us use Lemma~\ref{lemma:QLD} to construct a local decomposition for $\omega(O)$.
This local decomposition will involve a sequence of squares
\be
B_0=B \subset B_1 \subset B_2 \subset \ldots \subset \Lambda
\ee
such that $B_j\in \calS(q+2j)$ is the square obtained from $B$ by adding a boundary region
of thickness $j$ on all sides of $B$, that is,
\be
\label{B_j}
B_j=\{ u\in \Lambda\, : \, D(u,B)\le j\}.
\ee
Note that $B_j=\Lambda$ for large enough $j$.
Then we can write $\omega(O)$ as
\be
\label{omega(O)}
\omega(O)= D_{B_0}(O) + \sum_{j\ge 1} D_{B_j}(O), \quad
D_{B_0}(O) = \omega_B(O),
\quad D_{B_j}(O) = \omega_{B_j}(O) - \omega_{B_{j-1}}(O).
\ee
Note that $D_{B_j}(O)$ acts only on $B_j$, so Eq.~(\ref{omega(O)}) defines a local decomposition of $\omega(O)$.
Using Lemma~\ref{lemma:QLD} we infer that
\be
\| D_{B_j}(O)\| \le \| \omega_{B_j}(O) - \omega(O)\| + \| \omega_{B_{j-1}}(O) -\omega(O)\| \le
c (q+j)q^4 K^2  \cdot \| O\| \, e^{-\mu j}
\ee
%SM: updated the second bound above
for some constant $c$. In addition we have a standard bound (see for instance Ref.~\cite{BDLT08})
\be
\| D_{B_0}(O)\|  = \|\omega_B(O)\|  \le \frac12 \| [S_B, [S_B,O]] \| \le 2\|S_B\|^2 \|O\|.
\ee
Taking into account that
\be
\|S_B\| \le \sum_{p=2}^q \sum_{\substack{A\in \calS(p)\\A\subseteq B}} \| S_{p,A}\|
\le \sum_{p=2}^q (q-p)^2 K p^{-\alpha} e^{-\mu p} \le c q^2 K
\ee
for some constant $c$.
Thus for all $j\ge 0$ we have
\be
\label{Dbounds}
 \| D_{B_j}(O)\| \le cq^4(q+ j) K^2 \|O\| e^{-\mu j}.
\ee
%SM: updated the second bound above
 Having finished this warmup we can easily prove Lemma~\ref{lemma:second_order}.
\begin{proof}[\bf Proof of Lemma~\ref{lemma:second_order}]
Consider a local decomposition of $V$,
\be
V=\sum_{q\ge 2} \sum_{B\in \calS(q)} V_{q,B}, \quad \| V_{q,B}\| \le J q^{-\beta} e^{-\mu q}.
\ee
We construct a local decomposition of $\omega(V_{q,B})$ as in Eq.~(\ref{omega(O)}),
where $O\equiv V_{q,B}$. We get
\be
\omega(V)=\sum_{r\ge 2} \sum_{A\in \calS(r)} \Omega_{r,A},
\ee
where
%SBB1: a typo is fixed
\be
\Omega_{r,A} = \sum_{j=0}^{r/2-1} \; \; \sum_{\substack{B\in \calS(r-2j)\\ A=B_j}} D_{B_j}(V_{r-2j,B}).
\ee
Note that for fixed $j$ the sum over $B$ contains a single square $B$ such that $A=B_j$,
see Eq.~(\ref{B_j}).
Using Eq.~(\ref{Dbounds}) with $q=r-2j$, we arrive at:
\bea
\| \Omega_{r,A}\| &\le &  \sum_{j=0}^{r/2-1} \; \; \sum_{\substack{B\in \calS(r-2j)\\ A=B_j}}  K^2 \| V_{r-2j,B}\| (r-2j)^4 r   e^{-\mu j} \nn  \\
&\le & \sum_{j=0}^{r/2-1} K^2 J (r-2j)^{4-\beta} r  e^{-\mu(r-j)} \le  c K^2 r J e^{-\frac{\mu r}2}. \nn
\eea
%SM: above bound has changed, but now we have an extra r, which we don't like!
for some constant $c$ provided that $\beta\ge 6$.
We have shown that $\omega(V)$ is $(cK^2 J,\mu/2,-1)$-decaying.
\end{proof}

\section{Adiabatic continuation of logical operators}
\label{sec:continue}

\subsection{Dressed Operators}

Since the gap remains open up to a certain strength of perturbation, this means that the perturbed Hamiltonian
$H_0+V$ is adiabatically connected to the original Hamiltonian, $H_0$.  This adiabatic connection suggests that the perturbed
Hamiltonian should have similar properties to the unperturbed Hamiltonian.  For example, following~\cite{hwen}, we can define
string operators which have nontrivial action in the ground state subspace of the perturbed Hamiltonian and which have the correct commutation relations and expectation values.
In  fact, Theorem~\ref{thm:main}
will allow us to do even more, to define local operators that create defect excitations with
well-defined energies.

To construct these operators,
we use quasi-adiabatic continuation.
We define a continuous family of Hamiltonians,
\be
H_s=H_0+sV,
\ee
so that as $s$ varies from $0$ to $1$, $H_s$ continuously interpolates between $H_0$ and the perturbed Hamiltonian.

We define a quasi-adiabatic continuation operator,
${\cal D}_s$ by
\be
{\cal D}_s \equiv i\int {\rm d}t F(t) \exp(i H_s t) \Bigl( \partial_s H_s \Bigr) \exp(-i H_s t),
\ee
where the function $F(t)$ is defined to have the following properties.
First, the Fourier transform of $F(t)$, which we denote $\tilde F(\omega)$, obeys
\be
\label{correct}
|\omega| \geq 1/2 \quad \rightarrow \quad \tilde F(\omega)=-1/\omega.
\ee
Second, $\tilde F(\omega)$ is infinitely differentiable, so that $F(t)$ decays faster than any negative power of time for large $|t|$.
Third, $F(t)=-F(-t)$, so that ${\cal D}_s$ is anti-Hermitian.

We define a unitary operator $U_s$ by
\be
\label{Usdef}
U_s \equiv {\cal S}' \; \exp\Bigl\{ i\int_0^s {\rm d}s' {\cal D}_s \Bigr\},
\ee
where the notation ${\cal S}'$ denotes that the above equation (\ref{Usdef}) is an $s'$-ordered exponential.

The motivation for defining the above unitary operator is contained in the following lemmas:
\begin{lemma}
\label{lemma:samespace}
Let  $H_s$ be a differentiable family of Hamiltonians.
Let
$|\Psi^i(s)\rangle$ denote eigenstates of $H_s$ with energies $E_i$.
Let $E_{min}(s)<E_{max}(s)$ be continuous functions of $s$ and
\[
I(s)=\{ \lambda \in \RR\, : \,E_{min}(s)\le \lambda\le E_{max}(s)\}.
\]
Define a projector $P(s)$ onto an eigenspace of $H_s$ by
\be
P(s)=\sum_{i\, : \, E_i \in I(s)} \; |\Psi^i(s)\rangle \langle \Psi^i(s)|.
\ee
Assume that the
space that $P(s)$ projects onto is separated from the
rest of the spectrum by a gap of at least $1/2$ for all $s$ with $0\leq s \leq 1$.
That is, any eigenvalue  of $H_s$ either belongs to $I(s)$ or  is separated
from $I(s)$ by a gap at least $1/2$.
Then, for all $s$ with $0\leq s \leq 1$, we have
\be
\label{result}
P(s)=U_s P(0) U_s^\dagger.
\ee
\end{lemma}
\begin{proof}
By linear perturbation theory,
\begin{eqnarray}
\label{diff}
\partial_s P(s)&=&
\sum_{i\in I(s)} \;\; \sum_{j\notin I(s)} \; \;
\frac{1}{E_i-E_j} |\Psi^j(s)\rangle \Bigl(\langle \Psi^j(s)| \partial_s H(s) |\Psi^i(s)\rangle\Bigr) \langle \Psi^i(s)|+h.c
\\ \nonumber
&=&
-\sum_{i\in I(s)} \;\; \sum_{j\notin I(s)} \; \;
|\Psi^j(s)\rangle \Bigl(\langle \Psi^j(s)| \int{\rm d}t F(t) \exp(i H_s t)\partial_s H(s) \exp(-i H_s t)|\Psi^i(s)\rangle\Bigr) \langle \Psi^i(s)|+h.c \\ \nonumber
&=& i[{\cal D}_s,P_s].
\end{eqnarray}
The first equality in the above equation holds because
\begin{eqnarray}
&&\langle \Psi^j(s)| \int{\rm d}t F(t) \exp(i H_s t)\partial_s H(s) \exp(-i H_s t)|\Psi^i(s)\rangle
\\ \nonumber
&=&
\langle \Psi^j(s)| \int{\rm d}t F(t) \exp[i (E_j-E_i) t]\partial_s H(s)|\Psi^i(s)\rangle
\\ \nonumber
&=&
\tilde F(E_j-E_i)
\langle \Psi^j(s)|
 \partial_s H(s)|\Psi^i(s)\rangle,
\end{eqnarray}
where we use Eq.~(\ref{correct}) to show $\tilde F(E_j-E_i)=-1/(E_j-E_i)$ using the assumption on the gap in the spectrum that
$|E_j-E_i|\geq 1/2$.

Since $\partial_s (U_s P(0) U_s^{\dagger})=i[{\cal D}_s,U_s P(0) U_s^{\dagger}]$,
and $U_0=I$, Eq.~(\ref{result}) follows from Eq.~(\ref{diff}).
\end{proof}

This purpose for introducing this quasi-adiabatic continuation operator is to define certain ``dressed" operators, following the idea
introduced in \cite{hwen}.
Let $O_1,O_2,...$ be some operators that create defects when acting on the ground state of $H_0$.  These operators may be defined
to have certain commutation or anti-commutation requirements.  For example, if we have certain operators $O^E_i$ which create
electric defects on a given neighboring pair of sites, and operators $O^M_i$ which create magnetic defects on a given neighboring pair of plaquettes in a toric code state,
then these operators all commute with each other, except an electric and a magnetic operator anti-commute if the bond connecting the
sites and the bond connecting the plaquettes on the dual lattice intersect.
Then we define ``dressed" operators
$O_i(s)$ by
\be
O_i(s)=U_s O_i U_s^{\dagger}.
\ee
Since $U_s$ is unitary, these operators $O_i(s)$ obeys the same commutation or anti-commutation requirements:
\begin{eqnarray}
[O_i,O_j]=0 \rightarrow [O_i(s),O_j(s)]=0, \\ \nonumber
\{O_i,O_j\}=0 \rightarrow \{O_i(s),O_j(s)\}=0.
\end{eqnarray}

Let $P_n(s)$ project onto the eigenspace of $H(s)$ with energy in the interval $I_n$,
see Theorem~\ref{thm:main}. We assume that $J$ is chosen sufficiently small (see Section~\ref{subs:summary})
so that $I_n$  is separated from the rest of the spectrum by a gap at least $1/2$.
So, Lemma~\ref{lemma:samespace}  implies that  if a product of operators $O_1 O_2...O_m$ acting
any ground state of $H_0$ creates an eigenstate of $H_0$ with energy $n$, then the product of operators
$O_1(s) O_2(s) ... O_m(s)$ acting on any ground state state of $H_s$ creates a state with energy in the interval $I_n$.
To see this, note that a ground state $\Psi_0(s)$ of $H_s$ can be written as $U_s \Psi_0(0)$ for some ground state $\Psi_0(0)$ of $H_0$.
Then,
\begin{eqnarray}
P_n(s) O_1(s) ... O_m(s) | \Psi_0(s)\rangle &= &
U_s P_n(0) O_1(0) ... O_m(0) | \Psi_0\rangle \\ \nonumber
=U_s (0) O_1(0) ... O_m(0) | \Psi_0\rangle  &=&
O_1(s) ... O_m(s) | \Psi_0(s)\rangle.
\end{eqnarray}

Further, if some product  of operators $O_1 O_2 ...$ has a given expectation value in the ground state of $H_0$, then
the corresponding dressed operators have the same expectation value in the ground state of $H_s$.  For example, since a product of $\sigma^x$  around a contractable loop has expectation value unity in the ground state of the toric code, the
same product of dressed operators will have the same expectation value in the ground state of $H_s$.

The next important property we want to show is that the operators $O_i(s)$ are local.  To do these, we need a Lieb-Robinson bound for
quasi-adiabatic continuation.
Before describing this, some comments.
We have chosen to define the quasi-adiabatic continuation using ``exact" expressions.  That is, we have chosen a filter function
$F(t)$ such that its Fourier transform is exactly equal to $1/\omega$ outside a certain interval.  This is Osborne's\cite{tjo} modification of the quasi-adiabatic continuation
of \cite{hwen}.  As a result, our filter function $F(t)$ decays faster than any negative power of $t$, but does {\it not} decay exponentially
in $t$ (however, there do exist such filter functions with $|F(t)|$ decaying exponentially in a polynomial of $t$~\cite{inprep}).  This contrasts with the approach in \cite{hwen}, where an approximation was used that gave a filter function decaying
exponentially in $t^2$.  Each approach has certain advantages and disadvantages.  The approach in \cite{hwen} has the advantage
that one can often get exponentially good approximations, rather than approximations which merely decay faster than any power.
However, Osborne's modification which we use here has a few advantages also.  First, we get an exact result that $\Psi_0(s)=U_s\Psi_0(0)$ above, while
the approach in \cite{hwen} only gives approximate estimates.
Second, it is much easier to derive locality estimates.
In particular, the Lieb-Robinson bound for quasi-adiabatic continuation in the next lemma, which is the key step to prove locality of
the dressed operators,
is much easier than \cite{hwen,hall}.  The reason that this bound is so much simpler is the following:
the Fourier transform of the filter function here can be chosen to be some given bounded function of $\omega$.  In \cite{hwen,hall}, we have a parameter-dependent family of filter functions.
As this parameter is increased, the accuracy of the approximation improves,
but also the upper bound on the Fourier transform of the filter function gets worse, and hence the bound on the norm of ${\cal D}$
worsens.

The locality of the dressed operators depends on locality properties of $H_0,V$.  We require
that $H_0,V$ are both sums of
of local operators $H_{0,Z},V_Z$, where both $H_{0,Z},V_Z$ obey a bound that, for all sites $u\in \Lambda$,
\begin{eqnarray}
\label{HVloc}
\sum_{Z\ni u} \Vert H_{0,Z} \Vert |Z| \exp(\mu{\rm diam}(Z))=O(1), \\ \nonumber
\sum_{Z\ni u} \Vert V_Z \Vert |Z| \exp(\mu{\rm diam}(Z)) \leq J<\infty,
\end{eqnarray}
where $|Z|$ denotes the cardinality of $Z$, for some positive constants $\mu,J$.

Given such locality property, for any finite dimensional system
we have a Lieb-Robinson bound for $H_0+sV$ that, for any
operator $O_A$ supported on set $A$ and any $O_B$ supported on set $B$,
\be
\label{lrqad}
\| [\exp(i H_st) O_A \exp(-i H_s t),O_B] \| \leq \exp[-\mu ({\rm dist}(A,B)-v_{LR} t)] |A| \|O_A \| \|O_B \|,
\ee
where $v_{LR}$ is some constant which depends on $\mu,J$.  This bound is shown in \cite{hastings-koma}.
\begin{lemma}
Let $H_0$ and $V$ obey Eq.~(\ref{HVloc}).
Then, if $O_A$ is supported on set $A$ and $O_B$ is supported on set $B$,
\be
\| [U_s O_A U_s^\dagger,O_B] \| \leq h({\rm dist}(A,B)) |A| \| O_A \| \|O_B \|,
\ee
where $|A|$ denotes the cardinality of $A$, for $0\leq s \leq 1$,
for some function $h(l)$ which decays faster than any negative power of $l$.
Similarly,
\be
\| [U_s^\dagger O_A U_s,O_B] \| \leq h({\rm dist}(A,B)) |A| \| O_A \| \|O_B \|,
\ee

Further,
the operator
${\cal D}_s$ is a sum of operators ${\cal D}_s(Z)$, with
\be
\label{normbound}
\| {\cal D}_s(Z) \| \leq {\rm const}. \times \| V_Z  \|,
\ee
and with the property that each such operator
${\cal D}_s(Z)$ obeys, for any operator $O_B$,
\be
\| [{\cal D}_s(Z),O_B] \| \leq h'({\rm dist}(A,B)) |Z| \|V_Z \| \|O_B \|,
\ee
for some function $h'(l)$ which decays faster than any negative power of $l$.
\begin{proof}
We have
\be
{\cal D}_s = \sum_Z {\cal D}_s(Z),
\ee
where
\be
{\cal D}_s(Z)= i\int {\rm d}t F(t) \exp(i H_s t) V_Z \exp(-i H_s t).
\ee
Eq.~(\ref{normbound}) follows immediately from a triangle inequality because $\int {\rm d}t |F(t)|$ converges.

Let $O_B$ be an operator supported on set $B$.  Then
\begin{eqnarray}
\label{DlocBnd}
\| [{\cal D}_s(Z),O_B ] \| & \leq &
\int {\rm d}t |F(t)| \| [\exp(i H_s t) V_Z \exp(-i H_s t), O_B] \|.
\end{eqnarray}
For $t\leq {\rm dist}(Z,B)/2v_{LR}$, the above expression is exponentially small in ${\rm dist}(Z,B)$ by the
Lieb-Robinson bound, while
for larger $t$, the expression is bounded by $|F(t)|$, and hence this commutator decays faster thany
any negative power of ${\rm dist}(Z,B)$.

This decomposition and locality bound (\ref{DlocBnd}) implies that there is a Lieb-Robinson bound for evolution using ${\cal D}_s$ as a Hamiltonian.
Using the Lieb-Robinson bound in \cite{hastings-koma} (the use of this bound does require some geometric
properties of the lattice, which hold for any finite dimensional lattice),
we have that
\be
\label{lrb}
\| [U_s O_A U_s^\dagger,O_B] \| \leq \exp(c s) h'({\rm dist}(A,B)) |A| \| O_A \| \|O_B \|,
\ee
for some constant $c$ which depends on $F,J,\mu$ and on the geometric properties of the lattice, and for
some function $h'(l)$ which decays faster than any negative power of $l$.
Since we assume $s\leq 1$, Eq.~(\ref{lrqad}) follows from Eq.~(\ref{lrb}).
\end{proof}
\end{lemma}

Eq.~(\ref{lrqad}) implies that $O_A(s)=U_s O_A U_s^\dagger$ can be approximated by an operator $O_l$ localized on the set of
sites within distance $l$ of set $A$.  To see this, define $O_l=\int {\rm d}U U O_A(s) U^\dagger$, where the integral ranges over
unitary rotations, with Haar measure, supported on sites with distance greater than $l$ from set $A$.  Then, following \cite{bhv}, the desired result follows.

The string operators can be used to manipulate the ground states.  We will assume that the operators
$O_1$ which create defects are, in fact, unitaries.
For example, if a certain product of electric operators,
$O_1 O_2 ...$ creates a nontrivial loop around the torus and has nontrivial action on the ground state of $H_0$, then the
product of dressed operators $O_1(s) O_2(s) ...$ has the same action on the ground state of $H_s$.  However, one might
wonder how to create these dressed operators; it would be preferable not to have to solve quasi-adiabatic evolution equations
to calculate what the operators should be, and then to produce them by careful control of a time-dependent Hamiltonian.
Fortunately, we can use the gaps in the excited state spectrum to argue that it is possible to drag defects in the perturbed
Hamiltonian $H_s$ in an identical way to what is done in $H_0$.  First, we apply some local operator in an attempt to
create a defect pair on neighboring sites.  If this operator is not exactly equal to the desired dressed operator, then
this attempt may fail, in the sense that the energy may not fall into the desired range $I_2$,
see Theorem~\ref{thm:main}.  We can detect this failure
by measuring the energy, and cooling back to the ground state: that is, if there are extra defects created, we will drag them together, using the procedure described in the
next paragraph, to annihilate them.
Since the dressed operator is local, then some local operator (for example,
the undressed operator) is expected to have  non-vanishing matrix elements between
the ground state and the desired state.
Eventually we will succeed in creating the defect pair.

We now try to move the defect pair by weak changes in the Hamiltonian.  We
add a perturbation $U$ to the Hamiltonian with norm bounded by $1/4$.  Such a weak perturbation will keep the gap open to
the states with $3$ defects (of course, there are no such three defect states in the toric code, but in more general models there are).
So, by adding this perturbation and changing it adiabatically, we remain in the subspace with a defect pair.

In a toric code Hamiltonian, every state in this subspace is a linear combination of states which are given by
acting on a ground state $\Psi_0(s)$ by a linear  combination of products of strings of dressed operators.
Consider a given state created by a single string of dressed operators, with endpoints $i$ and $j$ which are far separated.  Call
this state $\Psi(S)$, where $S$ is a string with endpoints at $i,j$.
Note that using the property that any contractable loop of dressed operators has expectation value unity in the ground state
of $H_s$, we can show that we can deform these strings in any way that leaves the endpoints fixed while leaving the state unchanged.
Let us now see how the expectation value of local operators can change in the state $\Psi(S)$.  If $O$ is an operator
which is far separated from $i,j$, we claim that the expectation value of $O$ in state $\Psi(S)$ is close to its value in the ground
state.  To see these, note that that we can define a state $\Psi(S')=\Psi(S)$ where
$S'$ is a string which stays far from $O$.  So, without loss of generality, we
can assume that string $S$ is far from $O$.  Then, the using the locality
properties of the dressed operators, they almost commute with $O$, and so
\begin{eqnarray}
\langle \Psi(S)| O | \Psi(S) \rangle &=& \langle \Psi_0(s)|O_n^\dagger ...O_1^\dagger O O_1 ... O_n| \Psi_0(s) \rangle\\ \nonumber
&\approx &
\langle \Psi_0(s)|O_n^\dagger ... O_1^\dagger O_1 ... O_n O| \Psi_0(s) \rangle \\ \nonumber
&=& \langle \Psi_0(s)| O | \Psi_0(s) \rangle.
\end{eqnarray}

So, the perturbation $U$ can effect the state, but only if $U$ acts close to one of the endpoints of the string.  Further, if $U$ is
local, then one may show
that $U$ will have small matrix elements except between $\Psi(S)$ and $\Psi(S')$ for strings
$S'$ which differ from $S$ only by small motion of one of the endpoints.
Thus, if $U$ is chosen to have the property that it reduces the energy when an end of the string
is close to a certain point, the operator $U$ can indeed by used to drag the defects.

So, we have established that it is possible to create a defect pair and drag it.  The gap to the rest of the spectrum
prevents additional defects from being created.  If they are created, we can detect them by measuring the energy, and we can drag
them to destroy them and correct errors.

We can drag one of the defects around the system.  Since at every step we remain in a state which is a linear combination of states $\Psi(S)$, and eventually we succeed in dragging
a defect all
the way around the system, if we are able to return to a ground state $\Psi'_0(s)$, then
$\Psi'_0(s)$ is produced by
a local operation $O$ on a state $\Psi(S)$ where $S$ is a string with two nearby endpoints connected by a string that goes around the sample.  Any such
state $\Psi(S)$ is
a product of dressed operators acting on $\Psi_0(s)$.  So, it is equal to $\Psi(S)=U_s O_1...O_n \Psi_0(0)$.  Suppose $\Psi'_0(s)$ is a ground state such that
$\langle \Psi'_0(s), O \Psi(S) \rangle$
is non-negligible.  This expectation value is equal to
\be
\langle \Psi'_0(0)| \Bigl( U_s^\dagger O U_s \Bigr) O_1 ... O_n |\Psi_0(0) \rangle.
\ee
However, due to the locality properties of quasi-adiabatic continuation, the operator
$U_s^\dagger O U_s$ is approximately local.  Hence, the ground state $\Psi'_0(0)$ of
the unperturbed system is connected by a local operator to the state $O_1 ... O_n |\Psi_0(0) \rangle$.  So, the state $\Psi'_0(0)$ must be close to the state formed by acting on
$\Psi_0(0)$ with a noncontractible string that winds around the sample.  Thus, by
dragging the defects in the perturbed system, we succeed in effecting the desired
transformation on the ground state sector of the perturbed system.

One thing that can go wrong in this procedure is that we might inadvertently create the wrong type of defect.  We might accidentally
create a magnetic defect pair rather than an electric defect pair (one may show that the perturbation $U$ will have small amplitude
to change electric into magnetic defects if the defects are far separated and $U$ is local).  However, we expect to be able to locally tell the difference between these types of defects
and thus determine what kind of defect has been created, and, if the wrong type has
been created, to destroy the defect pair by bringing them together.

One would like to improve the arguments here to a full formal proof that using classical
control it is possible to correct local errors and perform controlled operations in this
perturbed toric code.  We will leave a complete proof of this to the future, but we
have established the existence of operators with the necessary properties and
of the gaps in the spectrum.

\section*{Acknowledgments}
We thank Barbara Terhal and David DiVincenzo for useful discussions.
Part of this work was done while SB and SM  were visiting
the Erwin Schr\"odinger International Institute for Mathematical Physics
at Vienna.
SB was partially  supported by the
DARPA QUEST program under contract number HR0011-09-C-0047.
SM thanks the organizers of the program on ``Quantum Information Science"
at the KITP at UC Santa Barbara, where part of this work was completed.  SM was supported
by NSF Grant DMS-07-57581 and DOE Contract DE-AC52-06NA25396.

\end{document}